%% file: probability_estimation.tex
\documentclass[10pt,aps,pra,twocolumn,longbibliography,superscriptaddress]{revtex4} 

\usepackage{amsmath,amsfonts,amssymb,amsthm}
\usepackage{mathtools}
\usepackage{setspace}
\usepackage[normalem]{ulem}
\usepackage{stmaryrd}
\usepackage{expdlist}

\usepackage{graphicx,xcolor}
\usepackage{hyperref}
\usepackage{url}

\usepackage{verbatim}
\usepackage{alltt}
\usepackage{moreverb}

\usepackage{xr} 

\allowdisplaybreaks

\input{xbytm_defs}

\begin{document}

\title{Certifying Quantum Randomness by Probability Estimation}

\author{Yanbao Zhang}\thanks{ 
                              yanbaoz@gmail.com}
\affiliation{NTT Basic Research Laboratories and NTT Research Center for Theoretical Quantum Physics, 
NTT Corporation, 3-1 Morinosato-Wakamiya, Atsugi, Kanagawa 243-0198, Japan}
\author{Emanuel Knill}
\affiliation{National Institute of Standards and Technology, Boulder, Colorado 80305, USA}
\affiliation{Center for Theory of Quantum Matter, University of Colorado, Boulder, Colorado 80309, USA}
\author{Peter Bierhorst}
\affiliation{National Institute of Standards and Technology, Boulder, Colorado 80305, USA}
\affiliation{Mathematics Department, University of New Orleans, New Orleans, Louisiana 70148, USA}

\begin{abstract}
  We introduce probability estimation, a broadly applicable framework to certify randomness in a finite sequence of measurement results without assuming that these results are independent and identically distributed. Probability estimation can take advantage of verifiable physical constraints, and the certification is with respect to classical side information.  Examples include randomness from single-photon measurements and device-independent randomness from Bell tests.  Advantages of probability estimation include adaptability to changing experimental conditions, 
  unproblematic early stopping when goals are achieved, optimal randomness rates, applicability to Bell tests with small violations, and unsurpassed finite-data efficiency. We greatly reduce latencies for producing random bits and formulate an associated rate-tradeoff problem of independent interest. We also show that the latency is determined by 
 an information-theoretic measure of nonlocality rather than the Bell violation. 
\end{abstract}

\maketitle
 
Randomness is a key enabling resource for computation and communication.
Besides being required for Monte-Carlo simulations and statistical
sampling, private random bits are needed for initiating authenticated
connections and establishing shared keys, both common tasks for
browsers, servers and other online entities~\cite{paar2010}.  Public
random bits from ``randomness beacons'' have applications to fair
resource sharing~\cite{fischer:qc2011a} and can seed private
randomness sources based on quantum mechanics~\cite{pironio:2013}.  
Common requirements for random bits are that they are unpredictable 
to all before they are generated, and private to the users 
before they are published. 

Quantum mechanics provides natural opportunities for generating
randomness. The best known example involves measuring a two-level
system that is in an equal superposition of its two levels.  A
disadvantage of such schemes is that they require trust in 
the measurement apparatus, and undiagnosed failures are
always a possibility. This disadvantage is overcome by a loophole-free
Bell test~\cite{colbeck:2007,colbeck:2011}, which can generate output
whose randomness can be certified solely by statistical tests of
setting and outcome frequencies. The devices preparing the quantum
states and performing the measurements may come from an untrusted source.  
This strategy for certified randomness generation is known as 
device-independent randomness generation (DIRG). 

Loophole-free Bell tests have been realized with nitrogen-vacancy (NV)
centers~\cite{hensen:2015}, with atoms~\cite{rosenfeld:qc2016a} and
with photons~\cite{giustina:2015, shalm:2015}, enabling the
possibility of full experimental implementations of DIRG. However, for
NV centers and atoms, the rate of trials is too low, and for photons,
the violation per trial is too small. As a result, previously
available DIRG protocols~\cite{pironio:2010, vazirani:2012,
  pironio:2013, fehr:2013, miller_c:qc2014a, miller_c:qc2014b,
  chung:2014, coudron:2014, arnon-friedman:2018, NietoSilleras:qc2016}
are not ready for implementation with current loophole-free Bell
tests.  These protocols do not achieve good finite-data 
efficiency and therefore require an impractical number of trials.  Experimental techniques will
improve, but for many applications of randomness generation, including 
randomness beacons and key generation, it is desirable to achieve 
finite-data efficiency that is as high as possible, since these applications 
often require short blocks of fresh random bits with minimum delay or latency. 

Excellent finite-data efficiency was achieved by a method that we described 
and implemented in Refs.~\cite{bierhorst:qc2017a,bierhorst:qc2018a}, which 
reduced the time required for generating $1024$ low-error random bits with 
respect to classical side information from hours to minutes for a
state-of-the-art photonic loophole-free Bell test.  The method in 
Refs.~\cite{bierhorst:qc2017a,bierhorst:qc2018a} is based on the 
prediction-based ratio (PBR) analysis~\cite{zhang:2011} for hypothesis 
tests of local realism. Specifically, in Refs.~\cite{bierhorst:qc2017a,bierhorst:qc2018a} 
we established a connection between the PBR-based $p$-value and the amount 
of randomness certified against classical side information.
The basis for success of the method of Refs.~\cite{bierhorst:qc2017a,bierhorst:qc2018a} 
motivates our development of probability estimation  for
randomness certification, with better finite-data efficiency and with
broader applications.  

In the probability estimation framework, the amount of certified randomness is \emph{directly} estimated 
without relying on hypothesis tests of local realism. To certify randomness, we 
first obtain a bound on the conditional probability of the observed outcomes 
given the chosen settings, valid for all classical side information. Then we  
show how to obtain conditional entropy estimates from this bound to quantify the 
number of extractable random bits~\cite{koenig:qc2009a}.  By focusing on
data-dependent probability estimates, we are able to take advantage of powerful 
statistical techniques to obtain the desired bound.  The statistical techniques 
are based on test supermartingales~\cite{shafer:qc2009a} and Markov's bounds. 
Probability estimation inherits several features of the theory of test supermartingales. 
For example, probability estimation has no independence or stationarity requirement on the 
probability distribution of trial results. Also, probability estimation supports stopping the 
experiment early, as soon as the randomness goal is achieved. 

Probability estimation is broadly applicable. In particular it is not limited to
device-independent scenarios and can be applied to traditional
randomness generation with quantum devices. 
Such applications are enabled by the notion of models,
which are sets of probability distributions that capture verified, 
physical constraints on device behavior. In the case of Bell
tests, these constraints include the familiar non-signaling 
conditions~\cite{PRBox,barrett:2005}. In the case of two-level 
systems such as polarized photons, the constraints can capture 
that measurement angles are within a known range, for example.

In this paper, we first describe the technical features 
of probability estimation and the main results that enable its practical use.  We propose a general
information-theoretic rate-tradeoff problem that closely relates to
finite-data efficiency.  We then show how the general theoretical
  concepts are instantiated in experimentally relevant examples
  involving Bell-test configurations. We demonstrate advantages of probability estimation
such as its optimal asymptotic randomness rates and show
large improvements in finite-data efficiency, which corresponds to 
great reductions in latency. 

\emph{Theory.}  Consider an experiment with ``inputs'' $\Sfnt{Z}$ and
``outputs'' $\Sfnt{C}$. The inputs normally consist of the random
choices made for measurement settings but may include choices of state
preparations such as in the protocols of Refs.~\cite{Lunghi:2015,
  VanHimbeeck:2017}.  The outputs consist of the corresponding
measurement outcomes. In the cases of interest, the inputs and outputs
are obtained in a sequence of $n$ time-ordered trials, where the $i$'th 
trial has input $Z_{i}$ and output $C_{i}$, and $\Sfnt{Z}=(Z_{i})_{i=1}^{n}$ 
and $\Sfnt{C}=(C_{i})_{i=1}^{n}$. We assume that $Z_{i}$ and $C_{i}$ are
countable-valued. We refer to the trial inputs and outputs collectively 
as the trial ``results'', and to the trials preceding 
the upcoming one as the ``past''.  The party with respect to which the randomness is 
intended to be unpredictable is represented by an external classical 
system, whose initial state before the experiment may be correlated 
with the devices used.  The classical system carries the side information 
$E$, which is assumed to be countable-valued.  After the experiment, 
the joint of $\Sfnt{Z}$, $\Sfnt{C}$ and $E$ is described by a probability distribution $\mu$. 
The upper-case symbols introduced in this paragraph are treated as random variables. 
As is conventional, their values are denoted by the corresponding lower-case 
symbols. 

The amount of extractable uniform randomness in $\Sfnt{C}$ conditional on both $\Sfnt{Z}$ and $E$ is
quantified by the (classical) smooth conditional min-entropy
$H_{\min}^{\epsilon}(\Sfnt{C}|\Sfnt{Z}E)_{\mu}$ where $\epsilon$ is
the ``error bound'' (or ``smoothness'') and $\mu$ is the joint distribution 
of $\Sfnt{Z}$, $\Sfnt{C}$ and $E$.  One way to define the smooth
conditional min-entropy is with the conditional guessing probability
$P_{\mathrm{guess}}(\Sfnt{C}|\Sfnt{Z}E)_{\mu}$ defined as the average
over values $\Sfnt{z}$ and $e$ of the maximum conditional probability
$\max_{\Sfnt{c}}\mu(\Sfnt{c}|\Sfnt{z}e)$.  The $\epsilon$-smooth
conditional min-entropy
$H_{\min}^{\epsilon}(\Sfnt{C}|\Sfnt{Z}E)_{\mu}$ is the greatest lower
bound of $-\log_{2}P_{\mathrm{guess}}(\Sfnt{C}|\Sfnt{Z}E)_{\mu'}$ for
all distributions $\mu'$ within total-variation distance $\epsilon$ of
$\mu$. Our goal is to obtain lower bounds on
$H_{\min}^{\epsilon}(\Sfnt{C}|\Sfnt{Z}E)_{\mu}$ with probability estimation.

The application of probability estimation requires a notion of models. A model $\cH$ 
for an experiment is defined as the set of all probability distributions
of $\Sfnt{Z}$ and $\Sfnt{C}$ achievable in the experiment conditionally on
 values $e$ of $E$.  If a joint distribution $\mu$ of $\Sfnt{Z}$, 
 $\Sfnt{C}$ and $E$ satisfies that for all $e$, the conditional distributions 
 $\mu(\Sfnt{C}\Sfnt{Z}|e)$, considered as distributions of $\Sfnt{Z}$ and $\Sfnt{C}$, 
 are in $\cH$, we say that the distribution $\mu$ satisfies the model $\cH$. 

To apply probability estimation to an experiment consisting of $n$ time-ordered trials, we construct 
the model $\cH$ for the experiment as a chain of models $\cC_{i}$ for each individual 
trial $i$ in the experiment. The trial model $\cC_{i}$ is defined as the set of 
all probability distributions of trial results $C_{i}Z_{i}$ achievable at 
the $i$'th trial conditionally on both the past trial results and the side 
information $E$. For example, for Bell tests, $\cC_{i}$ may be
the set of non-signaling distributions with uniformly random
inputs.  Let $\Sfnt{z}_{<i}=(z_{j})_{j=1}^{i-1}$ and
$\Sfnt{c}_{<i}=(c_{j})_{j=1}^{i-1}$ be the results before the
$i$'th trial. The sequences $\Sfnt{z}_{\leq i}$ and $\Sfnt{c}_{\leq
  i}$ are defined similarly.  The chained model $\cH$ consists of all 
  conditional distributions $\mu(\Sfnt{C}\Sfnt{Z}|e)$ satisfying the following 
  two conditions. First, at each trial $i$ the conditional distributions 
  $\mu(C_{i}Z_{i}|\Sfnt{c}_{<i}\Sfnt{z}_{<i}e)$ for all $\Sfnt{c}_{<i}$, 
  $\Sfnt{z}_{<i}$ and $e$ are in the trial model $\cC_{i}$.  
Second, at each trial $i$ the input $Z_i$ is independent of the past outputs $\Sfnt{C}_{<i}$ 
given $E$ and the past inputs $\Sfnt{Z}_{< i}$. The second condition prevents leaking information 
 about the past outputs through the future inputs, which is necessary for certifying
 randomness in the outputs $\Sfnt{C}$ conditional on both the inputs $\Sfnt{Z}$ and 
 the side information $E$.  In the common situation where the inputs are
chosen independently with distributions known before the experiment,
 the second condition is always satisfied.  

Since the model $\cH$ consists of all conditional distributions $\mu(\Sfnt{C}\Sfnt{Z}|e)$
regardless of the value $e$, the analyses in the next paragraph apply to the worst-case 
conditional distribution over $e$. To simplify notation we normally write  
the distribution $\mu(\Sfnt{C}\Sfnt{Z}|e)$ conditional on $e$ as $\mu_e(\Sfnt{C}\Sfnt{Z})$,
abbreviated as $\mu_{e}$. 

To estimate the conditional probability $\mu_e(\Sfnt{c}|\Sfnt{z})$, 
we design trial-wise probability estimation factors (PEFs) 
and multiply them. Consider a generic trial with trial model $\cC$, where 
for generic trials, we omit the trial index. Let $\beta>0$.  A PEF with 
power $\beta$ for $\cC$ is a function $F:cz\mapsto F(cz)\geq 0$ such that for all
$\sigma\in\cC$, $\Exp_{\sigma} \big(F(CZ)\sigma(C|Z)^{\beta}\big)\leq 1$,
where $\Exp$ denotes the expectation functional. 
Note that $F(cz)=1$ for all $cz$ defines a valid PEF with each positive power. 
For each $i$, let $F_{i}$ be a PEF with power $\beta$ for the $i$'th trial, 
where the PEF can be chosen adaptively based on the past results $\Sfnt{c}_{<i}\Sfnt{z}_{<i}$. 
Other information from the past may also be used, see
Ref.~\cite{knill:2017}. Let $T_{0}=1$ and $T_{i}=\prod_{j=1}^{i}F_{j}(C_{j}Z_{j})$. 
 The final value $T_{n}$ of the running product $T_i$, where $n$ is the total 
number of trials in the experiment, determines the probability estimate. 
 Specifically, for each value $e$ of $E$, each $\mu_e$ in the chained model 
 $\cH$, and $\epsilon>0$, 
 we have 
\begin{equation}
  \Prob_{\mu_e}\big(\mu_{e}(\Sfnt{C}|\Sfnt{Z}) \ge U(\Sfnt{C}\Sfnt{Z})\big)\leq \epsilon,
  \label{eq:prob_est}
\end{equation}
 where $\Prob_{\mu_e}$ denotes the probability according to the distribution $\mu_e$
and $U(\Sfnt{C}\Sfnt{Z})=(\epsilon T_n)^{-1/\beta}$. The proof of Eq.~\eqref{eq:prob_est} 
is given in Appendix~\ref{sec:PE_and_PEFs}. 
The meaning of Eq.~\eqref{eq:prob_est} is as follows: For each $e$ and each $\mu_e\in\cH$, 
the probability that $\Sfnt{C}$ and $\Sfnt{Z}$ take values $\Sfnt{c}$ and $\Sfnt{z}$ for 
which $U(\Sfnt{C}=\Sfnt{c},\Sfnt{Z}=\Sfnt{z})\leq\mu_e(\Sfnt{C}=\Sfnt{c}|\Sfnt{Z}=\Sfnt{z})$
is at most $\epsilon$. This defines $U(\Sfnt{C}\Sfnt{Z})=(\epsilon T_n)^{-1/\beta}$ 
as a level-$\epsilon$ probability estimator.

A main theorem of probability estimation is the connection between probability estimators and 
conditional min-entropy estimators, which is formalized as follows:
\begin{theorem} \label{thm:smooth_min_entropy_bound} Suppose that 
the joint distribution $\mu$ of $\Sfnt{Z}$, $\Sfnt{C}$ and $E$ satisfies the chained model $\cH$. 
Let  $1\geq \kappa, \epsilon>0$ and $1\geq  p \geq 1/|\Rng(\Sfnt{C})|$, where $|\Rng(\Sfnt{C})|$
is the number of possible outputs.  Define $\{\phi\}$ to be the event 
that $T_n \geq 1/(p^\beta\epsilon)$, and let $\kappa\leq \Prob_{\mu}(\phi)$.   
  Then the smooth conditional min-entropy satisfies
  \begin{equation*}
    H_{\min}^{\epsilon}(\Sfnt{C}|\Sfnt{Z}E;\phi) \geq
    -\log_2(p/\kappa^{1+1/\beta}). 
    \label{eq:smooth_min_entropy_bound}
  \end{equation*}
\end{theorem}
The probability of the event $\{\phi\}$ can be interpreted as the  
probability that the experiment succeeds, and $\kappa$ is an assumed 
lower bound on the success probability.  The theorem is proven in
Appendix~\ref{sec:minentropy_extraction}. 

When constructing PEFs, the power $\beta>0$ must be decided \emph{before} the
experiment and cannot be adapted.
Thm.~\ref{thm:smooth_min_entropy_bound} requires that $p$, $\epsilon$
and $\kappa$ also be chosen \emph{beforehand}, and success of the experiment
requires $T_n \geq 1/(p^\beta\epsilon)$, or equivalently,
\begin{equation} \label{eq:success_con}
\log_{2}(T_{n})/\beta + \log_{2}(\epsilon)/\beta\geq -\log_{2}(p).
\end{equation}
Since $\log_{2}(T_{n})=\sum_{i}\log_{2}(F_{i})$, \emph{before} the experiment 
we choose PEFs in order to aim for large expected values of the logarithms of the PEFs $F_{i}$.  
Consider a generic next trial with results $CZ$ and model $\cC$. Based on prior
calibrations or the frequencies of observed results in past trials, we
can determine a distribution $\nu\in\cC$ that is a good approximation 
to the distribution of the next trial's results $CZ$. Many experiments are designed 
so that each trial's distribution is close to $\nu$. The PEF can be 
optimized for this distribution but, by definition, is valid regardless 
of the actual distribution of the next trial in $\cC$. Thus, one way to 
optimize PEFs \emph{before} the next trial is as follows:
\begin{equation}
  \begin{array}[b]{ll}
    \textrm{Max:}& 
    \Exp_{\nu} \big(n\log_2(F(CZ))/\beta+\log_2(\epsilon)/\beta\big) \\ 
    \textrm{With:}& \sum_{cz}F(cz)\sigma(c|z)^{\beta}\sigma(cz) \leq 1
    \textrm{\ for all $\sigma\in\cC$}, \\
    & F(cz)\geq 0, \textrm{\ for all $cz$}.
  \end{array}\label{eq:opt_pef}
\end{equation}
The objective function is strictly concave and the constraints are
linear, so there is a unique maximum, which can be found by convex
programming.  More details are available in Appendix~\ref{subsec:optimization}. 

Before the experiment, one can also optimize the objective function in 
Eq.~\eqref{eq:opt_pef} with respect to the power $\beta$. During the experiment 
$\epsilon$ and $\beta$ are fixed, so it suffices to maximize
$\Exp_{\nu}\big(\log_{2}(F(CZ))\big)$. If during the experiment, the
running product $T_{i}$ with $i<n$ exceeds the target $1/(p^\beta\epsilon)$, 
we can set future PEFs to $F(CZ)=1$, which is a valid PEF with power $\beta$. 
This ensures that $T_{n}=T_{i}$ and is equivalent to stopping the experiment after 
trial $i$. Since the target needs to be set conservatively 
in order to make the actual experiment succeed with high probability, 
this can result in a significant reduction in the number of trials actually executed.

A question is how PEFs perform asymptotically for a stable experiment. 
This question is answered by determining the rate per trial of entropy 
production assuming constant $\epsilon$ and $\kappa$ independent of the number of trials.  
In view of Thm.~\ref{thm:smooth_min_entropy_bound}, after $n$ trials the entropy rate 
is given by $\big(-\log_2(p)+\log_2(\kappa^{1+1/\beta})\big)/n$. Considering 
Eq.~\eqref{eq:success_con}, when $n$ is large the entropy rate is dominated by 
$\log_{2}(T_{n})/(n\beta)$, which is equal to $\sum_{i=1}^{n}\log_{2}(F_{i})/(n\beta)$.  
Therefore, if each trial has distribution $\nu$ and each trial model is the same $\cC$, 
then in the limit of large $n$ the asymptotic entropy rate witnessed by a PEF $F$ 
with power $\beta$ is given by $\Exp_{\nu}\big(\log_2(F(CZ))/\beta\big)$. Define the rate
\begin{equation}
  g(\beta)=\sup_{F}\Exp_{\nu} \big(\log_2(F(CZ))/\beta\big),
  \label{eq:gain_rate}
\end{equation}
where the supremum is over PEFs $F$ with power $\beta$ for $\cC$. 
The maximum asymptotic entropy rate at constant $\epsilon$ and $\kappa$ 
witnessed by PEFs is $g_{0}=\sup_{\beta>0} g(\beta)$.  The rate
$g(\beta)$ is non-increasing in $\beta$ (see Appendix~\ref{subsec:pef_property}), 
so $g_{0}$ is determined by the limit as $\beta$ goes to zero.  A theorem proven in
Ref.~\cite{knill:2017} is that $g_{0}$ is the worst-case
conditional entropy $H(C|ZE)$ over joint distributions of $CZE$
allowed by $\cC$ with marginal $\nu$.  Since this is a tight upper
bound on the asymptotic randomness rate~\cite{tomamichel:qc2009a}, 
probability estimation is asymptotically optimal and we identify $g_{0}$ as the asymptotic
  randomness rate.  We also remark that probability estimation enables exponential expansion 
  of input randomness~\cite{knill:2017}. 

For finite data and applications requiring fresh blocks of randomness,
the rate $g_{0}$ is not achieved. To understand why, consider
the problem of certifying a fixed number of bits $b$ of randomness at
error bound $\epsilon$ and with as few trials as possible, where each trial
has distribution $\nu$.  In view of
Thm.~\ref{thm:smooth_min_entropy_bound}, the PEF optimization problem
in Eq.~\eqref{eq:opt_pef}, and the definition of $g(\beta)$ in
Eq.~\eqref{eq:gain_rate}, $n$ needs to be sufficiently large so that
\begin{equation} 
  n g(\beta)+\log_{2}(\epsilon)/\beta+(1+1/\beta)\log_{2}(\kappa)\geq b.
  \label{eq:critical_cond}
\end{equation}
The left-hand side is maximized at positive $\beta$, whereas
$g(\beta)$ increases to $g_{0}$ as $\beta$ goes to zero. As a result the
best actual rate $b/n$ is less than $g_{0}$. 

Setting $\kappa=1$ in Eq.~\eqref{eq:critical_cond} shows that the number
of trials $n$ must exceed $-\log_{2}(\epsilon)/(\beta g(\beta))$
before randomness can be produced, which suggests that the maximum of
$\beta g(\beta)$ is a good indicator of finite-data performance.
Another way to arrive at this quantity is to consider
$\epsilon=2^{-\gamma n}$, where $\gamma>0$ is the ``certificate
rate''. Given $\nu$ and the trial model, we can ask for the maximum
certificate rate for which it is possible to have positive entropy
rate at $\kappa=1$.  It follows from Eq.~\eqref{eq:critical_cond}
with $\kappa=1$ that this rate is at most 
\begin{equation}\label{eq:certificate_rate}
\gamma_{\textrm{PEF}} =\sup_{\beta>0}\beta g(\beta).  
\end{equation} 
We propose a general information-theoretic rate-tradeoff problem given trial model 
$\cC$ and $\nu\in\cC$: For a given certificate rate $\gamma$, determine the
supremum of the entropy rates achievable by protocols. Eq.~\eqref{eq:critical_cond} 
implies lower bounds on the resulting tradeoff curve.  

Our protocol assumes classical-only side information.  There are more
costly DIRG protocols that handle quantum side
information~\cite{vazirani:2012, miller_c:qc2014a, miller_c:qc2014b,
  chung:2014, coudron:2014, arnon-friedman:2018}, but verifying that
side information is effectively classical only requires confirming
that the quantum devices used in the experiment have no long-term 
quantum memory.  Verifying the absence of long-term quantum memory 
in current experiments is possibly less difficult than ensuring  
that there are no backdoors or information leaks in the experiment's  
hardware and software. 

\emph{Applications.}  We consider DIRG with the standard two-party,
two-setting, two-outcome Bell-test
configuration~\cite{clauser:qc1969a}. The parties are labeled
$\Pfnt{A}$ and $\Pfnt{B}$.  In each trial, a source prepares a
  state shared between the parties, and each party chooses a random
setting (their input) and obtains a measurement outcome (their
output).  We write $Z=XY$, where $X$ and $Y$ are the inputs of
$\Pfnt{A}$ and $\Pfnt{B}$, and $C=AB$, where $A$ and $B$ are the
respective outputs.  For this configuration, $A, B, X, Y\in\{0,1\}$.

Consider the trial model $\cN$ consisting of distributions of $ABXY$
with uniformly random inputs and satisfying non-signaling~\cite{PRBox}.  
 We begin by determining and comparing the asymptotic randomness rates 
 witnessed by different methods. The rates are usually quantified as functions 
 of the expectation $\hat{I}$ of the CHSH Bell function (Eq.~\ref{eq:CHSH_function}) 
 for $\hat{I}>2$ (the classical upper bound).  
We prove in Appendix~\ref{subsec:gain_rate} that the maximum asymptotic randomness 
rate for any $\nu\in\cN$ is equal to $(\hat{I} -2)/2$, and the rate $g_{0}$ 
witnessed by PEFs matches this value. Most previous studies, such as 
Refs.~\cite{pironio:2013, pironio:2010, fehr:2013, Acin:2012, nieto:2014, 
bancal:2014,  NietoSilleras:qc2016}, estimate the asymptotic randomness 
rate by the single-trial conditional min-entropy $H_{\min}(AB|XYE)$.  We
determine that $H_{\min}(AB|XYE)=-\log_2((6-\hat I)/4)<g_0$ when
$2<\hat I<4$.  As $\hat{I}$ decreases to $2$ the ratio of $g_0$ to 
$H_{\min}(AB|XYE)$ approaches $1.386$,
demonstrating an improvement at small violations.

Next, we investigate finite-data performance. We consider three
different families of quantum-achievable distributions of trial results.  
For the first family $\nu_{\textrm{E},\theta}$, $\Pfnt{A}$ and $\Pfnt{B}$
share the unbalanced Bell state $|\Psi_{\theta}\rangle=\cos\theta
|00\rangle+\sin\theta |11\rangle$ with $\theta\in(0,\pi/4]$ and apply
projective measurements that maximize $\hat I$. This determines
$\nu_{\textrm{E},\theta}$.  This family contains the
  goal states for many experiments suffering from detector
  inefficiency.  For the second family $\nu_{\textrm{W},p}$,
$\Pfnt{A}$ and $\Pfnt{B}$ share a Werner state $\rho=p
|\Psi_{\pi/4}\rangle \langle \Psi_{\pi/4}|+(1-p)\one/4$ with $p
\in(1/\sqrt{2},1]$ and again apply measurements that maximize $\hat
I$.  Werner states are standard examples in quantum
  information and are among the worst states for our application.  In
experiments with photons, measurements are implemented with imperfect
 detectors.  For the third family $\nu_{\textrm{P},\eta}$,
$\Pfnt{A}$ and $\Pfnt{B}$ use detectors with efficiency $\eta\in
(2/3,1)$ to implement the measurements and to close the detection
loophole~\cite{Eberhard1993}.  They choose the unbalanced Bell state
$|\Psi_{\theta}\rangle$ and measurements such that an 
information-theoretic measure of nonlocality, the statistical
strength for rejecting local realism~\cite{vanDam:2005, acin2005, zhang:2010}, 
is maximized.

\begin{figure}
  \begin{center}
    \resizebox{!}{1.5in}{\includegraphics[viewport=5.5cm 10.5cm 16.5cm 18.5cm]{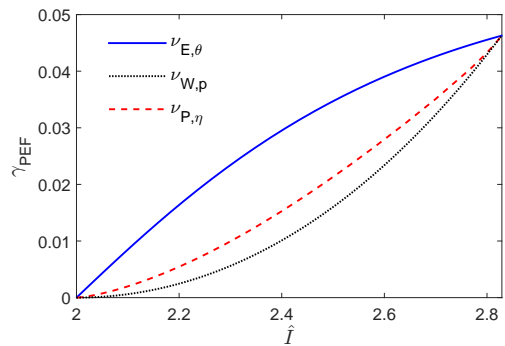}}    
  \end{center}
  \caption{Maximum certificate rates $\gamma_{\textrm{PEF}}$ (Eq.~\eqref{eq:certificate_rate}) as a 
  function of $\hat{I}$ for each family of distributions. }
  \label{fig:certificate_rate}
\end{figure}

For each family of distributions, we determine the maximum certificate rate 
$\gamma_{\textrm{PEF}}$ as given in Eq.~\eqref{eq:certificate_rate}. 
For this, we consider the trial model $\cN$, but we note that 
$\gamma_{\textrm{PEF}}$ does not depend on the specific constraints on the 
quantum-achievable conditional distributions $\Prob(AB|XY)$ (see Appendix~\ref{subsec:certificate_rate}).  
As an indicator of finite-data performance, $\gamma_{\textrm{PEF}}$ depends not only on $\hat I$, 
but also on the distribution $\nu$. To illustrate this behavior, we plot the rates 
$\gamma_{\textrm{PEF}}$ as a function of $\hat I$ for each family of distributions 
in Fig.~\ref{fig:certificate_rate}. To obtain these plots, we note that $\hat I$ 
is a monotonic function of the parameter $\theta$, $p$ or $\eta$ for each family. 
We also find that $\gamma_{\textrm{PEF}}$ is given by the statistical 
strength of the distribution $\nu$ for rejecting local realism (see Appendix~\ref{subsec:certificate_rate} for a proof).  
Conventionally, experiments are designed to maximize $\hat{I}$, but in general, 
the optimal state and measurements maximizing $\hat{I}$ are different from those 
maximizing the statistical strength~\cite{acin2005,zhang:2010}.

We further determine the minimum number of trials, $n_{\textrm{PEF},b}$, required to certify $b$ 
bits of $\epsilon$-smooth conditional min-entropy with a given distribution $\nu$ of trial results. 
From Eq.~\eqref{eq:critical_cond}, we get 
\begin{equation*}
n_{\textrm{PEF},b}=\inf_{\beta>0} \frac{b\beta-\log_{2}(\epsilon)-(1+\beta)\log_2(\kappa)}{\beta g(\beta)},
\label{eq:trial_bound}
\end{equation*}
where for simplicity we allow non-integer values for $n_{\textrm{PEF},b}$.  We can upper bound 
$n_{\textrm{PEF},b}$ by means of the simpler-to-compute certificate rate $\gamma_{\textrm{PEF}}$ 
given in Eq.~\eqref{eq:certificate_rate}.  
For the trial model $\cN$, $\gamma_{\textrm{PEF}}$ is achieved when $\beta$ is above a threshold 
$\beta_0$ that depends on $\nu$ (see Appendix~\ref{subsec:certificate_rate}). 
From $\gamma_{\textrm{PEF}}$ and $\beta_0$, we can determine the upper bound 
\begin{equation*}
n'_{\textrm{PEF},b}=\big(b\beta_0-\log_{2}(\epsilon)-(1+\beta_0)\log_2(\kappa)\big)/\gamma_{\textrm{PEF}}
\label{eq:estimated_trial_bound}
\end{equation*}
on $n_{\textrm{PEF},b}$.  The minimum number 
  of trials required can be determined for other published protocols, which
usually certify conditional min-entropy from $\hat I$.  (An exception
is Ref.~\cite{NietoSilleras:qc2016} but the minimum number of trials
required is worse.)  We consider the protocol ``PM'' of
Ref.~\cite{pironio:2013} and the entropy accumulation protocol ``EAT''
of Ref.~\cite{arnon-friedman:2018}.  From Thm.~1 of
Ref.~\cite{pironio:2013} with $\kappa=1$ and $b\searrow 0$, we obtain
a lower bound
\begin{equation*}
  n_{\textrm{PM},0}=-2\log_{e}(\epsilon)/\big((\hat{I}-2)/(4+2\sqrt{2}) \big)^2.
  \label{eq:PM_trial_bound}
\end{equation*}
For the EAT protocol, we determine an explicit lower bound
$n_{\textrm{EAT},b}$ in Appendix~\ref{subsec:eat}. This lower bound applies for 
$b\geq 0$ and $\epsilon, \kappa\in(0,1]$, and is valid with respect to 
quantum side information for the trial model consisting of  
quantum-achievable distributions.

We compare the three protocols over a broad range of $\hat I$ for
$b\searrow 0$, $\epsilon=10^{-6}$, and $\kappa=1$.  For each family  
of distributions above, we compute the improvement factors
given by $f_{\textrm{PM}}=n_{\textrm{PM},0}/n'_{\textrm{PEF},0}$ and
$f_{\textrm{EAT}}=n_{\textrm{EAT},0}/n'_{\textrm{PEF},0}$.  For
$\nu_{\textrm{W},p}$, the improvement factors depend weakly on $\hat I$:
$f_{\textrm{PM}}$ increases from $3.89$ at $\hat I = 2.008$ to $4.36$
at $\hat{I}=2\sqrt{2}$, while $f_{\textrm{EAT}}$ increases from
$84.97$ at $\hat I = 2.008$ to $86.35$ at $\hat I=2\sqrt{2}$.  For
$\nu_{\textrm{E},\theta}$ and $\nu_{\textrm{P},\eta}$, the improvement
factors can be much larger and depend strongly on $\hat I$, monotonically
decreasing with $\hat I$ as shown in Fig.~\ref{fig:improvement}.  The
improvement is particularly notable at small violations which are
typical in current photonic loophole-free Bell tests.  We remark 
that similar comparison results were obtained with other choices of 
the values for $\epsilon$ and $\kappa$.

\begin{figure}
  \begin{center}
     \resizebox{!}{1.5in}{\includegraphics[viewport=5.5cm 10.5cm 16.5cm 18.5cm]{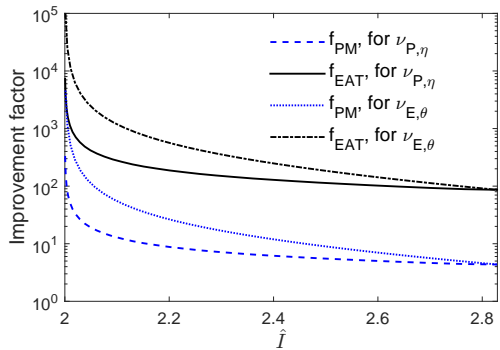}}
  \end{center}
  \caption{Improvement factors as a function of $\hat{I}$. 
  }
  \label{fig:improvement}
\end{figure}

The large latency reduction with probability estimation persists for certifying blocks of 
randomness.  For randomness beacons, good reference values are $b=512$ 
and $\epsilon=2^{-64}$.  We also set $\kappa=2^{-64}$. Setting
$\kappa=\epsilon$ is a common conservative choice, but we remark that
 soundness for randomness generation can be defined with a 
 better tradeoff between $\epsilon$ and $\kappa$~\cite{knill:2017}.  
 We consider the trial model $\cT$ of distributions with uniformly 
 random inputs, satisfying both non-signaling conditions~\cite{PRBox} 
 and Tsirelson's bounds~\cite{Tsirelson:1980}. 
Consider the state-of-the-art photonic loophole-free Bell test reported  
in Ref.~\cite{bierhorst:qc2018a}.  With probability estimation, the number of
trials required for the distribution inferred from the measurement
statistics is $4.668\times 10^{7}$, which would require about $7.78$
minutes of running time in the referenced experiment.  With entropy 
accumulation~\cite{arnon-friedman:2018}, $2.887\times 10^{11}$ trials 
taking $802$ hours would be required. For atomic experiments, we can use 
the distribution inferred from the measurement statistics in Ref.~\cite{rosenfeld:qc2016a}, 
for which probability estimation requires $7.354\times 10^{4}$ trials, while 
entropy accumulation~\cite{arnon-friedman:2018} requires $5.629\times 10^{6}$. 
The experiment of Ref.~\cite{rosenfeld:qc2016a} observed $1$ to $2$ trials per minute, 
so probability estimation would have needed at least $612.8$ hours of data collection, which while 
impractical is still less than the $5.35$ years required by entropy 
accumulation~\cite{arnon-friedman:2018}. 

Finally, we briefly discuss the performance of probability estimation on DIRG with published 
Bell-test experimental data. 
The first experimental demonstration of conditional min-entropy certification 
for DIRG is reported in Ref.~\cite{pironio:2010}. The method therein certifies 
the presence of $42$ random bits at error bound $\epsilon=10^{-2}$ against 
classical side information, where the trial model consists of quantum-achievable 
distributions with uniform inputs. (The lower bound of the protocol 
success probability $\kappa=1$ was used implicitly in Ref.~\cite{pironio:2010}, 
so $\kappa=1$ in the following comparison.)  For the same data but with the less 
restrictive trial model $\cT$, 
probability estimation certifies the presence of at least nine times 
more random bits with $\epsilon=10^{-2}$. With $\epsilon=10^{-6}$ probability estimation can still certify
the presence of $80$ random bits, while other methods fail to certify any random 
bits.  For the loophole-free Bell-test data reported in Ref.~\cite{shalm:2015} and
analyzed in our previous work Ref.~\cite{bierhorst:qc2017a}, the presence of 
$894$ random bits at $\epsilon=10^{-3}$ was certified against classical side 
information with the trial model $\cN$. 
Further, $256$ private random bits within $10^{-3}$ 
(in terms of the total-variation distance) of uniform were extracted in Ref.~\cite{bierhorst:qc2017a}.  
With probability estimation we can certify the presence of approximately two times more random bits 
at $\epsilon=10^{-3}$. The presence of four times more bits can be certified if we 
use the more restrictive trial model $\cT$. Furthermore, we can certify randomness 
even when the input distribution is not precisely known, which was an issue in
the experiment of Ref.~\cite{shalm:2015}.  Applications to other experimental 
distributions, complete analyses of the mentioned experiments, and details on 
handling input choices whose probabilities are not precisely known are in 
Ref.~\cite{knill:2017}. 

In conclusion, probability estimation is a powerful and flexible framework for certifying
randomness in data from a finite sequence of experimental trials.
Implemented with probability estimation factors, it witnesses optimal asymptotic randomness
rates.  For practical applications requiring fixed-size blocks of
random bits, it can reduce the latencies by orders of
magnitude even for high-quality devices.  Latency is a notable problem 
for device-independent quantum key generation (DIQKD). If probability estimation can be 
extended to accommodate security against quantum side information, the 
latency reductions may be extendable to DIQKD 
by means of existing constructions~\cite{arnon-friedman:2018}.

Finally we remark that if the trial results 
are explainable by local realism, no device-independent randomness would be 
certified by probability estimation. The reason is as follows.  For simplicity
 we assume that the input distribution is fixed and known~\footnote{The argument 
can be generalized to the case that the input distribution is not 
precisely known after considering the construction of the corresponding PEFs detailed in Ref.~\cite{knill:2017}.}.
Consider a generic trial with results $CZ$ and model $\cC$. 
Let $\cP_{\text{LR}}$ be the set of distributions of $CZ$ explainable by local realism, 
which is a convex polytope with a finite number of extremal distributions 
$\sigma_{\text{LR},k}$, $k=1,2,...,K$. Since $\cP_{\text{LR}}$ is a subset of $\cC$, 
by definition a PEF $F$ with power $\beta$ satisfies the condition 
\begin{equation}\label{eq:lr_cons}
\sum_{cz}F(cz)\sigma_{\text{LR},k}(c|z)^{\beta}\sigma_{\text{LR},k}(cz) \leq 1,
\end{equation} 
for each $k$.  For each extremal distribution $\sigma_{\text{LR},k}$ in $\cP_{\text{LR}}$ 
and each $cz$, the value of $\sigma_{\text{LR},k}(c|z)$ is either $0$ or $1$, from which it follows that 
$\sigma_{\text{LR},k}(c|z)^{\beta}\sigma_{\text{LR},k}(cz)=\sigma_{\text{LR},k}(cz)$. 
 Eq.~\eqref{eq:lr_cons} now becomes  
\begin{equation}\label{eq:lr_cons2}
\Exp_{\sigma_{\text{LR},k}}\big(F(CZ)\big)=\sum_{cz}F(cz)\sigma_{\text{LR},k}(cz) \leq 1. 
\end{equation}
Since any local realistic distribution can be written as a convex mixture of extremal 
distributions $\sigma_{\text{LR},k}$, $k=1,2,...,K$, Eq.~\eqref{eq:lr_cons2} implies that 
for all distributions $\nu\in\cP_{\text{LR}}$ 
\begin{equation}\label{eq:pbr_def}
\Exp_{\nu}\big(F(CZ)\big)\leq 1. 
\end{equation}
By the concavity of the logarithm function and Eq.~\eqref{eq:pbr_def} we get that 
\begin{equation}
\Exp_{\nu} \big(\log_2(F(CZ))\big)\leq \log_2\big(\Exp_{\nu}(F(CZ))\big) \leq 0. \notag
\end{equation}
Hence, the asymptotic entropy rate in Eq.~\eqref{eq:gain_rate} 
cannot be positive if the distribution of trial results is explainable by local realism.  
Furthermore, Eq.~\eqref{eq:pbr_def} shows
that the PEF $F$ is a test factor for the hypothesis test of local 
realism~\cite{zhang:2011} (see Appendix~\ref{subsec:supermart} for the formal 
definition of test factors). So, if a finite sequence of trial results 
 is  explainable by local realism and $F_{i}$ is a PEF with power 
$\beta$ for the $i$'th trial, according to Ref.~\cite{zhang:2011} 
the success event $T_n \geq 1/(p^\beta\epsilon)$ with 
$T_{n}=\prod_{i=1}^{n}F_{i}$ in Thm.~\ref{thm:smooth_min_entropy_bound} 
for randomness certification would happen with probability 
at most $p^\beta\epsilon$.

\begin{acknowledgments}
  We thank D. N. Matsukevich for providing the
  experimental data for Ref.~\cite{pironio:2010}, Bill Munro, Carl
  Miller, Kevin Coakley, and Paulina Kuo for help with reviewing this paper.  
  This work includes contributions of the National Institute of Standards
  and Technology, which are not subject to U.S. copyright.
\end{acknowledgments}

\input{probability_estimation.bbl}

\clearpage
\newpage
\onecolumngrid

\appendix
\section*{Appendix}
\label{sect:appendix}

\section{Notation}

Much of this work concerns stochastic sequences, that is, sequences of
random variables (RVs). RVs are functions on an underlying probability
space.  The range of an RV is called its value space and may be
  thought of as the set of its observable values or realizations.
Here, all RVs have countable value spaces. We truncate sequences of
RVs so that we only consider finitely many RVs at a time.  With this
the underlying probability space is countable too. We use upper-case
letters such as $A,B,\ldots,X,Y,\ldots$ to denote RVs.  The value
space of an RV such as $X$ is denoted by $\Rng(X)$. The cardinality of
the value space of $X$ is $|\Rng(X)|$.  Values of RVs are denoted by
the corresponding lower-case letters. Thus $x$ is a value of $X$,
often thought of as the particular value realized in an experiment.  
 When using symbols for values of RVs, they are implicitly assumed to be
members of the range of the corresponding RV.  In many cases, the
value space is a set of letters or a set of strings of a given length.
We use juxtaposition to denote concatenation of letters and strings.
 Stochastic sequences are denoted by capital bold-face letters, with the 
 corresponding lower-case bold-face letters for their values. For example, 
 we write $\Sfnt{A}=(A_{i})_{i=1}^{N}$ and $\Sfnt{A}_{\le m}=(A_{i})_{i=1}^{m}$.
Our conventions for indices are that we generically use $N$ to denote
a large upper bound on sequence lengths, $n$ to denote the available
length and $i,j,k,l,m$ as running indices.  By convention,
$\Sfnt{A}_{\le 0}$ is the empty sequence of RVs. Its value is 
constant.   When multiple stochastic sequences
are in play, we refer to the collection of $i$'th RVs in the sequences
as the data from the $i$'th trial. We typically imagine the trials as
happening in time and being performed by an experimenter.  We refer to
the data from the trials preceding the upcoming one as the
``past''. The past can also include initial conditions and any
additional information that may have been obtained. These are normally
implicit when referring to or conditioning on the past.

Probabilities are denoted by $\Prob(\ldots)$.  If there are multiple
probability distributions involved, we disambiguate with a subscript
such as in $\Prob_{\nu}(\ldots)$ or simply $\nu(\ldots)$, where $\nu$
is a probability distribution. We generally reserve the symbol $\mu$
for the global, implicit probability distribution, and may write
$\mu(\ldots)$ instead of $\Prob(\ldots)$ or $\Prob_{\mu}(\ldots)$. 
Expectations are similarly
denoted by $\Exp(\ldots)$ or $\Exp_{\mu}(\ldots)$.  If $\phi$ is a
logical expression involving RVs, then $\{\phi\}$ denotes the event
where $\phi$ is true for the values realized by the RVs.  For example,
$\{f(X)>0\}$ is the event $\{x:f(x)>0\}$ written in full
set notation.  The brackets $\{\ldots\}$ are omitted for events inside
$\Prob(\ldots)$ or $\Exp(\ldots)$. As is conventional, commas
separating logical expressions are interpreted as conjunction.  When
the capital/lower-case convention can be unambiguously interpreted, we
abbreviate ``$X=x$'' as ``$x$''. For example, with this convention,
$\Prob(x,y)=\Prob(X=x,Y=y)$.  Furthermore, we omit commas in the
abbreviated notation, so $\Prob(xy)=\Prob(x,y)$.  RVs or functions of
RVs appearing outside an event but inside $\Prob(\ldots)$ or after the
conditioner in $\Exp(\ldots|\ldots)$ result in an expression that is
itself an RV. We can define these without complications because of our
assumption that the event space is countable.  Here are two
examples. $\Probv(f(X)|Y)$ is  a function of the RVs
$X$ and $Y$ and can be described as the RV whose value is
$\Prob(f(X)=f(x)|Y=y)$ whenever the values of $X$ and $Y$ are $x$ and
$y$, respectively. Similarly $\Exp(X|Y)$ is the RV defined as a
function of $Y$, with value $\Exp(X|Y=y)$ whenever $Y$ has value $y$.
Note that $X$ plays a different role before the conditioners in
$\Exp(\ldots)$ than it does in $\Prob(\ldots)$,
as $\Exp(X|Y)$ is not a function of $X$, but only of $Y$.  We comment
that conditional probabilities with conditioners having probability
zero are not well-defined, but in most cases can be defined
arbitrarily. Typically, they occur in a context where they are
multiplied by the probability of the conditioner and thereby
contribute zero regardless. An important context involves
expectations, where we use the convention that when expanding an
expectation over a set of values as a sum, zero-probability
values are omitted. We do so without explicitly adding the constraints
to the summation variables.  We generally use conditional
probabilities without explicitly checking for probability-zero
conditioners, but it is necessary to monitor for well-definedness of
the expressions obtained.

To denote general probability distributions, usually on the joint
value spaces of RVs, we use symbols such as $\mu,\nu,\sigma$, 
with modifiers as necessary.  As mentioned, we reserve the unmodified $\mu$
for the distinguished global distribution under consideration, if
there is one. Other symbols typically refer to probability
distributions defined on the joint range of  a subset of the
available RVs.  We usually just say ``distribution'' instead of
``probability distribution''.  The terms ``distributions on
$\Rng(X)$'' and ``distributions of $X$'' are synonymous.  
 If $\nu$ is a joint distribution of RVs, then we extend the conventions 
 for arguments of $\Prob(\ldots)$ to arguments of $\nu$, as long as all 
 the arguments are determined by the RVs for which $\nu$ is defined. 
 For example, if $\nu$ is a joint distribution of $X$, $Y$, and $Z$, 
 then $\nu(x|y)$ has the expected meaning, as does the RV $\nu(X|Y)$ 
 in contexts requiring no other RVs. Further, $\nu(X)$ and $\nu(XY)$
 are the marginal distributions of $X$ and $XY$, respectively, according 
 to $\nu$.  

In our work, probability distributions are constrained by a ``model'', 
which is defined as a set of distributions and denoted by letters such 
as $\cH$ or $\cC$. The models for trials to be considered here are 
usually convex and closed.  

The total-variation (TV) distance between $\nu$ and $\nu'$ is defined
as
\begin{equation}
  \TV(\nu,\nu') = \sum_{x}(\nu(x)-\nu'(x))\knuth{\nu(x)\geq\nu'(x)}
  = \frac{1}{2}\sum_{x}|\nu(x)-\nu'(x)|,
  \label{eq:def_tv}
\end{equation}
where $\knuth{\phi}$ for a logical expression $\phi$ denotes the
$\{0,1\}$-valued function evaluating to $1$ iff $\phi$ is true.  True
to its name, the TV distance satisfies the triangle inequality. Here
are three other useful properties: First, if $\nu$ and $\nu'$ are joint 
distributions of $X$ and $Y$ and the marginals satisfy $\nu(Y)=\nu'(Y)$, 
then the TV distance between $\nu$ and $\nu'$ is the average of the TV 
distances of the $Y$-conditional distributions:
\begin{equation}\label{eq:tv_samemarg}
 \TV(\nu,\nu')=\sum_{y}\nu(y) \TV(\nu(X|y),\nu'(X|y)).
\end{equation}
Second, if for all $y$, the conditional distributions $\nu(X|y)=\nu'(X|y)$, then the 
TV distance between $\nu$ and $\nu'$ is given by the TV distance between the marginals on $Y$:
\begin{equation}\label{eq:tv_samecond}
\TV(\nu,\nu') = \TV(\nu(Y),\nu'(Y)).
\end{equation}
Third, the TV distance satisfies the data-processing inequality. That
is, for any stochastic process $\cE$ on $\Rng(X)$ and distributions
$\nu$ and $\nu'$ of $X$, $\TV(\cE(\nu),\cE(\nu'))\leq \TV(\nu,\nu')$.
We use this property only for functions $\cE$, but for general forms 
of this result, see Ref.~\cite{pardo:1997}.  The above properties 
of TV distances are well known, specific proofs can be found in 
Refs.~\cite{knill:2017, bierhorst:qc2018a}.

When constructing distributions close to a given one in TV distance, 
which we need to do for the proof of Thm.~\ref{thm:smooth_min_entropy_bound} 
in the main text,  
it is often convenient to work with subprobability distributions.  A
subprobability distribution of $X$ is a sub-normalized non-negative
measure on $\Rng(X)$, which in our case is simply a non-negative
function $\tilde\nu$ on $\Rng(X)$ with weight
$w(\tilde\nu)=\sum_{x}\tilde\nu(x)\leq 1$. For expressions not
involving conditionals, we use the same conventions for subprobability
distributions as for probability distributions.  When comparing
subprobability distributions, $\tilde\nu\leq\tilde\nu'$ means that
for all $x$, $\tilde\nu(x)\leq\tilde\nu'(x)$, and we say that
$\tilde\nu'$ ``dominates'' $\tilde\nu$.

\begin{lemma}\label{lm:tvdist_weight}
  Let $\tilde\nu$ be a subprobability distribution of $X$ of weight
  $w=1-\epsilon$.  Let $\nu$ and $\nu'$ be distributions of $X$
  satisfying $\tilde\nu\leq\nu$ and $\tilde\nu\leq \nu'$. Then
  $\TV(\nu,\nu')\leq \epsilon$.
\end{lemma}

\begin{proof}
  Calculate 
  \begin{align}
    \TV(\nu,\nu') &=
    \sum_{x}(\nu(x)-\nu'(x))\knuth{\nu(x)\geq \nu'(x)}\notag\\
    &\leq \sum_{x}(\nu(x)-\tilde\nu(x))\knuth{\nu(x)\geq \tilde\nu(x)}\notag\\
    &=\sum_{x}(\nu(x)-\tilde\nu(x))\notag\\
    &=1-w=\epsilon.\notag
  \end{align}
\end{proof}

\begin{lemma}\label{lm:tvsub}
  Assume that $p\geq1/|\Rng(X)|$.
  Let $\nu$ be a distribution of $X$ and $\tilde\nu\leq\nu$ a subprobability
  distribution of $X$ with weight $w=1-\epsilon$
  and $\tilde\nu\leq p$. Then there exists a distribution $\nu'$ of $X$
  with $\nu'\geq \tilde\nu$, $\nu'\leq p$, and $\TV(\nu,\nu')\leq \epsilon$.
\end{lemma}

\begin{proof}
  Because $p\geq 1/|\Rng(X)|$, that is, $\sum_{x}p \geq 1$, 
  and for all $x$, $\tilde\nu(x)\leq p$, there exists a distribution
  $\nu'\geq\tilde\nu$ with $\nu'\leq p$. Since $\nu'$ and $\nu$ are
  distributions dominating $\tilde\nu$ and by
  Lem.~\ref{lm:tvdist_weight}, $\TV(\nu,\nu')\leq \epsilon$.
\end{proof}

\section{Test Supermartingales and Test Factors}
\label{subsec:supermart}

\begin{definition}
  A \emph{test supermartingale~\cite{shafer:qc2009a} with respect to a
    stochastic sequence $\Sfnt{R}$ and model $\cH$} is a stochastic
  sequence $\Sfnt{T}=(T_{i})_{i=0}^{N}$ with the properties that
  1) $T_{0}=1$, 2) for all $i$ $T_{i}\ge 0$, 3) $T_{i}$ is determined by
  $\Sfnt{R}_{\leq i}$ and the governing distribution, and 4) for all
  distributions in $\cH$, $\Exp(T_{i+1}|\Sfnt{R}_{\le i})\leq T_{i}$.
  The ratios $F_{i}=T_{i}/T_{i-1}$ with $F_{i}=1$ if $T_{i-1}=0$ are
  called the \emph{test factors of $\Sfnt{T}$}.
\end{definition}

Here $\Sfnt{R}$ captures the relevant information that accumulates in
a sequence of trials. It does not need to be accessible to the
experimenter.  Between trials $i$ and $i+1$, the sequence
$\Sfnt{R}_{\le i}$ is called the past. In the definition, we allow for
$T_{i}$ to depend on the governing distribution $\mu$.  With this, for
a given $\mu$, $T_{i}$ is a function of $\Sfnt{R}_{\leq i}$.  Below,
when stating that RVs are determined, we implicitly include the
possibility of dependence on $\mu$ without mention.  The
$\mu$-dependence can arise through expressions such as
$\Exp_{\mu}(G|\Sfnt{R}_{\le i})$ for some $G$, which is determined by
$\Sfnt{R}_{\leq i}$ given $\mu$.  One way to formalize this is to
consider $\mu$-parameterized families of RVs. We do not make this
explicit and simply allow for our RVs to be implicitly parameterized
by $\mu$. We note that the governing distribution in a given
experiment or situation is fixed but usually unknown with most of its
features inaccessible. As a result, many RVs used in mathematical
arguments cannot be observed even in principle. Nevertheless, they
play important roles in establishing relationships between observed
and inferred quantities.

Defining $F_{i}=1$ when $T_{i-1}=0$ makes sense because given
$\{T_{i-1}=0\}$, we have $\{T_{i}=0\}$ with probability $1$.  The
sequence $\Sfnt{F}=(F_{i})_{i=1}^{N}$ satisfies the conditions that for all $i$, 1)
$F_{i}\geq 0$, 2) $F_{i}$ is determined by $\Sfnt{R}_{\le i}$, and 3)
for all distributions in $\cH$, $\Exp(F_{i+1}|\Sfnt{R}_{\le i})\leq 1$.  
We can define test
supermartingales in terms of such sequences: Let $\Sfnt{F}$ be a
stochastic sequence satisfying the three conditions.  Then the
stochastic sequence with members $T_0=1$ and $T_{i}=\prod_{1\leq j\leq
  i}F_{j}$ for $i\geq1$ is a test supermartingale. It suffices to
check that $\Exp(T_{i+1}|\Sfnt{R}_{\leq i}) \leq T_{i}$. This follows
from
\begin{equation*}
  \Exp(T_{i+1}|\Sfnt{R}_{\leq i})=
  \Exp(F_{i+1}T_{i}|\Sfnt{R}_{\leq i})
  = \Exp(F_{i+1}|\Sfnt{R}_{\leq i})T_{i}\leq T_{i},
\end{equation*}
where we pulled out the determined quantity $T_{i}$ from the
conditional expectation.  In this work, we construct test
supermartingales from sequences $\Sfnt{F}$ with the above
properties. We refer to any such sequence as a sequence of test
factors, without necessarily making the associated test
supermartingale explicit. We extend the terminology by calling an RV
$F$ a test factor with respect to $\cH$ if $F\geq 0$ and
$\Exp(F)\leq 1$ for all distributions in $\cH$. Note that 
$F=1$ is a valid test factor.

For an overview of test supermartingales and their properties, see
Ref.~\cite{shafer:qc2009a}. The notion of test supermartingales
and proofs of their basic properties are due to
Ville~\cite{ville:qc1939a} in the same work that introduced the
notion of martingales.  The name ``test supermartingale''
appears to have been introduced in Ref.~\cite{shafer:qc2009a}. Test
supermartingales play an important theoretical role in proving many
results in martingale theory, including that of proving tail bounds
for large classes of martingales. They have been studied and applied
to Bell tests~\cite{zhang:2011,zhang:2013,christensen:qc2015a}.

The definition implies that for a test supermartingale $\Sfnt{T}$, for
all $n$, $\Exp(T_{n})\leq 1$. This follows inductively from
$\Exp(T_{i+1})=\Exp(\Exp(T_{i+1}|\Sfnt{R}_{\le i}))\le \Exp(T_i)$ and
$T_0=1$. An application of Markov's inequality shows that for all
$\epsilon>0$,
\begin{equation}
  \Prob(T_{n} \geq 1/\epsilon)\leq \epsilon.
  \label{eq:testmart_markov}
\end{equation}
Thus, a large final value $t=T_{n}$ of the test supermartingale is
evidence against $\cH$ in a hypothesis test with $\cH$ as the
(composite) null hypothesis.  Specifically, the RV $1/T$ is a
$p$-value bound against $\cH$, where in general, the RV $U$ is a
$p$-value bound against $\cH$ if for all distributions in $\cH$,
$\Prob(U\leq\epsilon)\leq\epsilon$. 

One can produce a test supermartingale adaptively by determining the
test factors $F_{i+1}$ to be used at the next trial.  If the $i$'th
trial's data is $R_{i}$, including any incidental information
obtained, then $F_{i+1}$ is expressed as a function of $\Sfnt{R}_{\le i}$
and data from the $(i+1)$'th trial (a ``past-parameterized'' function
of $R_{i+1}$), and constructed to satisfy $F_{i+1}\geq 0$ and
$\Exp(F_{i+1}|\Sfnt{R}_{\le i})\leq 1$ for any distribution in the
model $\cH$.  Note that inbetween trials, we can effectively stop the
experiment by assigning all future $F_{i+1}=1$, which is a valid test factor, 
conditional on the past. This is equivalent to constructing the stopped process relative
to a stopping rule. This argument also shows that the stopped process
is still a test supermartingale.  

More generally, we use test supermartingales for estimating lower
bounds on products of positive stochastic sequences $\Sfnt{G}$. Such
lower bounds are associated with unbounded-above confidence
intervals. We need the following definition:

\begin{definition} 
  Let $U,V,X$ be RVs and $1\ge\epsilon\ge 0$.  $I=[U,V]$ is a
  \emph{confidence interval for $X$ at level $\epsilon$
    with respect to $\cH$} if for all distributions in $\cH$ 
    we have $\Prob(U\leq X\leq V)\geq 1-\epsilon$. The quantity
    $\Prob(U\leq X\leq V)$ is called the \emph{coverage probability}.
\end{definition}

As noted above, the RVs $U$, $V$ and $X$ may be $\mu$-dependent.  For
textbook examples of confidence intervals such as in Ch. 2.4.3 of Ref~\cite{Shao},
$X$ is a parameter determined by $\mu$, and $U$ and $V$ are 
 obtained according to a known distribution for an estimator of $X$.  
 The quantity $\epsilon$ in the definition is a significance level, 
 which corresponds to a confidence level of $(1-\epsilon)$. 
 The following technical lemma will be used in the next section. 

\begin{lemma} \label{lem:confidence_interval}
  Let $\Sfnt{F}$ and $\Sfnt{G}$ be two stochastic sequences with
  $F_{i}\in[0,\infty)$, $G_{i}\in (0,\infty]$, and $F_{i}$ and
  $G_{i}$ determined by $\Sfnt{R}_{\leq i}$.  Define $T_{0}=1$,
  $T_{i}=\prod_{1\leq j \leq i}F_{i}$ and $U_{0}=1$,
  $U_{i}=\prod_{1\leq j\leq i}G_{i}$, and suppose that for all
  $\mu\in\cH$, $\Exp(F_{i+1}/G_{i+1}|\Sfnt{R}_{\le i}) \leq 1$.
  Then $[T_{n}\epsilon,\infty)$ is a confidence interval for $U_{n}$
  at level $\epsilon$ with respect to $\cH$. 
\end{lemma}

\begin{proof}
  The assumptions imply that the sequence $(F_{i}/G_{i})_{i=1}^{N}$ forms
  a sequence of test factors with respect to $\cH$ and generate the test
  supermartingale $\Sfnt{T}/\Sfnt{U}$, where division in this
  expression is term-by-term.  Therefore, by Eq.~\eqref{eq:testmart_markov},
  \begin{equation}
    \Prob(T_{n}\epsilon \geq U_{n}) = \Prob(T_{n}/U_{n}\geq 1/\epsilon) \leq \epsilon,
    \label{eq:itestmart_markov}
  \end{equation}
  so $[T_{n}\epsilon,\infty)$ is a confidence
  interval for $U_{n}$ at level $\epsilon$.  
\end{proof}

\section{Proof of Main Results}
In this section, we show how to perform probability estimation and how 
to certify smooth conditional min-entropy by probability estimation. 

\subsection{Probability Estimation by Test Supermartingales: Proof of Main Text Eq.~\eqref{eq:prob_est}}
\label{sec:PE_and_PEFs}

We consider the situation where $\Sfnt{CZ}$ is a  time-ordered sequence of $n$ trial
results, and the classical side information is represented by an RV $E$ with 
countable value space. In an experiment, $\Sfnt{Z}$ and $\Sfnt{C}$ are the 
inputs and outputs of the quantum devices, and the side information $E$ is carried 
by an external classical system $\Pfnt{E}$. Before the experiment, the initial state of
$\Pfnt{E}$ may be correlated with the quantum devices. At each trial of the experiment, 
we allow arbitrary one-way communication from the system $\Pfnt{E}$ to the devices.
For example, $\Pfnt{E}$ can initialize the state of the quantum devices  
via a one-way communication channel. We also allow the possibility that the device 
initialization at a trial by $\Pfnt{E}$ depends on the past inputs preceding the trial.  
This implies that the random inputs $\Sfnt{Z}$ can come from public-randomness sources, as first 
pointed out in Ref.~\cite{pironio:2013}.  
However, at any stage of the experiment the information of the outputs $\Sfnt{C}$ 
cannot be leaked to the system $\Pfnt{E}$. After the experiment, we observe $\Sfnt{Z}$ and $\Sfnt{C}$, 
but not the side information $E$.

A model $\cH$ for an experiment is defined as the set of joint probability 
distributions of $\Sfnt{CZ}$ that satisfy the known constraints and
 consists of all achievable probability distributions of $\Sfnt{CZ}$  
conditional on values $e$ of $E$.  Thus we say that a joint distribution $\mu$ of
$\Sfnt{CZ}$ and $E$ satisfies the model $\cH$ 
 if $\mu(\Sfnt{C}\Sfnt{Z}|E=e)\in\cH$ for each value $e$.

We focus on probability estimates with lower bounds on coverage probabilities 
that do not depend on $E$. Our specific goal is to prove Eq.~\eqref{eq:prob_est}   
in the main text. We will show that the probability bound of 
$U(\Sfnt{CZ})=(T_n\epsilon)^{-1/\beta}$ in Eq.~\eqref{eq:prob_est} of the main text %
is an instance of what we call an ``$\epsilon$-uniform probability estimator'':
\begin{definition}
  Let $1\geq\epsilon\geq 0$. The function $U:\Rng(\Sfnt{CZ})\rightarrow [0,\infty)$ is a
  \emph{level-$\epsilon$ $E$-uniform probability estimator for $\cH$}
  ($\epsilon$-$\UPE$ or with specifics, $\epsilon$-$\UPE(\Sfnt{C}|\Sfnt{Z}E;\cH)$) 
  if for all $e$ and distributions $\mu$ satisfying the model $\cH$, we have
  $\Prob_{\mu}(U(\Sfnt{CZ})\geq\mu(\Sfnt{C}|\Sfnt{Z}e)|e)\geq 1-\epsilon$.  We omit 
  specifics such as $\cH$ if they are clear from context. 
  \label{def_UPE}
\end{definition}

  
We can obtain $\epsilon$-$\UPE$s by constructing test supermartingales. In order to   
achieve this goal, we consider models $\cH(\cC)$ of distributions of $\Sfnt{CZ}$  
 constructed from a chain of trial models $\cC_{i+1|\Sfnt{c}_{\le i}\Sfnt{z}_{\le i}e}$,  where the trial model $\cC_{i+1|\Sfnt{c}_{\le i}\Sfnt{z}_{\le i}e}$ is defined as the set of all achievable distributions
of $C_{i+1}Z_{i+1}$ conditional on both the past results $\Sfnt{c}_{\le i}\Sfnt{z}_{\le i}$ and the value $e$ of $E$.  
The chained model $\cH(\cC)$ consists of all conditional distributions $\mu(\Sfnt{C}\Sfnt{Z}|e)$ 
satisfying the following two properties.  
First, for all $i$,  $\Sfnt{c}_{\le i}\Sfnt{z}_{\le i}$, and $e$, the conditional distributions
\begin{equation*}
  \mu(C_{i+1}Z_{i+1}|\Sfnt{c}_{\le i}\Sfnt{z}_{\le i}e)\in \cC_{i+1|\Sfnt{c}_{\le i}\Sfnt{z}_{\le i}e}.
\end{equation*}
Second, the joint distribution $\mu$ of $\Sfnt{CZ}$ and $E$ satisfies that $Z_{i+1}$ is independent of 
$\Sfnt{C}_{\le i}$ conditionally on both $\Sfnt{Z}_{\le i}$ and $E$. 
The second condition is needed in order to be able to estimate 
$\Sfnt{Z}E$-conditional probabilities of $\Sfnt{C}$ and corresponds
to the Markov-chain condition in the entropy accumulation
framework~\cite{arnon-friedman:2018}.  

In many cases, the trial models $\cC_{i+1|\Sfnt{c}_{\le i}\Sfnt{z}_{\le i}e}$
do not depend on the past outputs $\Sfnt{c}_{\le i}$, but probability estimation can take advantage of 
dependence on the past inputs $\Sfnt{z}_{\le i}$. Such dependence captures 
the possibility that at the $(i+1)$'th trial the device initialization by 
the external classical system $\Pfnt{E}$ depends on the past inputs $\Sfnt{z}_{\le i}$. 
In applications involving Bell-test configurations, the trial models
capture constraints on the input distributions and on non-signaling
or quantum behavior of the devices.  For simplicity,
we write $\cC_{i+1}=\cC_{i+1|\Sfnt{c}_{\le i}\Sfnt{z}_{\le i}e}$,
leaving the conditional parameters implicit. Normally, models for individual 
trials $\cC_{i+1}$ are convex and closed.  If they are not, we note that 
our results generally extend to the convex closures of the  trial models used. 

For chained models $\cH(\cC)$, we can construct $\epsilon$-$\UPE$s  from products 
of ``probability estimation factors'' according to the following definition, 
see also the paragraph containing Eq.~\eqref{eq:prob_est} in the main text. 
\begin{definition}
  \label{def:pef}
  Let $\beta>0$, and let $\cC$ be any model, not necessarily convex.
  A \emph{probability estimation factor (PEF) with power $\beta$ for
    $\cC$} is a non-negative RV $F=F(CZ)$ such that for all
  $\sigma\in\cC$, $\Exp_{\sigma}(F\sigma(C|Z)^{\beta})\leq 1$.
\end{definition}
We emphasize that a PEF is a function of the trial results $CZ$, but not 
of the side information $E$. 

Consider the model $\cH(\cC)$ constructed as a chain of trial models $\cC_{i}$. 
Let $F_{i}$ be PEFs with power $\beta>0$ for $\cC_{i}$, past-parameterized 
by $\Sfnt{C}_{< i}$ and $\Sfnt{Z}_{< i}$. Define $T_{0}=1$, 
$T_{i}=\prod_{1\leq j\leq i}F_{j}$ for $i\geq 1$, and 
\begin{equation}
  U(\Sfnt{CZ}) =(T_{n}\epsilon)^{-1/\beta}. \label{eq:upe_final}
\end{equation}
Then, $U(\Sfnt{CZ})$ satisfies the inequality in Eq.~\eqref{eq:prob_est} 
of the main text as proven in the following theorem, and is therefore an 
$\epsilon$-$\UPE$. To simplify notation in the following theorem, we 
normally write the distribution $\mu(\Sfnt{C}\Sfnt{Z}|e)$ conditional on 
$e$ as $\mu_e(\Sfnt{C}\Sfnt{Z})$, abbreviated as $\mu_{e}$.  
\begin{theorem}\label{thm:uest_constr}
  Fix $\beta>0$.  For each value $e$ of $E$, each $\mu_{e} \in \cH(\cC)$, and $\epsilon>0$, 
  the following inequality  holds:
\begin{equation} \label{eq:thm:uest_constr:pefupe}
 \Prob_{\mu_e}\big(\mu_{e}(\Sfnt{C}|\Sfnt{Z}) \ge (\epsilon T_n)^{-1/\beta}\big)\leq \epsilon.
\end{equation}
\end{theorem}

Note that $\beta$ cannot be adapted during the trials.  On the other
hand, before the $i$'th trial, we can design the PEFs $F_{i}$ for
the particular constraints relevant to the $i$'th trial.

\begin{proof}
  We first observe that for each value $e$ of $E$, 
  \begin{equation} \label{eq:prob_de_compos}
    \prod_{j=0}^{i-1}\mu_e(C_{j+1}|Z_{j+1}\Sfnt{Z}_{\le j}\Sfnt{C}_{\le j}) =
    \mu_e(\Sfnt{C}_{\leq i}|\Sfnt{Z}_{\leq i}).
  \end{equation} 
  This follows by induction with the identity
  \begin{align}
    \mu_e(\Sfnt{C}_{\le j+1}|\Sfnt{Z}_{\le j+1}) &=
    \mu_e(C_{j+1}|Z_{j+1}\Sfnt{Z}_{\le j}\Sfnt{C}_{\le j})
    \mu_e(\Sfnt{C}_{\le j}|Z_{j+1}\Sfnt{Z}_{\le j})\notag\\
    &= \mu_e(C_{j+1}|Z_{j+1}\Sfnt{Z}_{\le j}\Sfnt{C}_{\le j})
    \mu_e(\Sfnt{C}_{\le j}|\Sfnt{Z}_{\le j}) \notag
  \end{align}
  by conditional independence of $Z_{j+1}$ on $\Sfnt{C}_{\le j}$ given 
  $\Sfnt{Z}_{\le j}$ and $E=e$.

  We claim that for each $e$, $F_{i+1}\mu_e(C_{i+1}|Z_{i+1}\Sfnt{Z}_{\leq
    i}\Sfnt{C}_{\leq i})^{\beta}$ is a test factor determined by
  $\Sfnt{C}_{\leq i+1}\Sfnt{Z}_{\leq i+1}$.  To prove this claim, for all
  $\Sfnt{c}_{\leq i}\Sfnt{z}_{\leq i}$, the distributions
  $\nu=\mu_e(C_{i+1}Z_{i+1}|\Sfnt{c}_{\leq i}\Sfnt{z}_{\leq
    i})\in\cC_{i+1}$.  With
  $F_{i+1}=F_{i+1}(C_{i+1}Z_{i+1};\Sfnt{c}_{\leq i}\Sfnt{z}_{\leq
    i})$, we obtain the bound
  \begin{align}
    \Exp\left(F_{i+1}\mu_e(C_{i+1}|Z_{i+1}\Sfnt{z}_{\leq
        i}\Sfnt{c}_{\leq i})^{\beta}|\Sfnt{c}_{\leq i}\Sfnt{z}_{\leq i}\right)
    &=\Exp_{\nu}\left(F_{i+1}\nu(C_{i+1}|Z_{i+1})^{\beta}\right)\notag\\
    &\leq 1, \notag
  \end{align}
  where we invoked the assumption that $F_{i+1}$ is a PEF with
  power $\beta$ for $\cC_{i+1}$.  By arbitrariness of
  $\Sfnt{c}_{\leq i}\Sfnt{z}_{\leq i}$, and because the factors 
  $F_{i+1}\mu_e(C_{i+1}|Z_{i+1}\Sfnt{Z}_{\leq i}\Sfnt{C}_{\leq i})^{\beta}$ 
  are determined by $\Sfnt{C}_{\leq i+1}\Sfnt{Z}_{\leq i+1}$, the claim
  follows.  The product of these test factors is
  \begin{align}
    \prod_{j=0}^{i-1}F_{j+1}\mu_e(C_{j+1}|Z_{j+1}\Sfnt{Z}_{\leq
      j}\Sfnt{C}_{\leq j})^{\beta} &= T_{i}\prod_{j=0}^{i-1}\mu_e(C_{j+1}|Z_{j+1}\Sfnt{Z}_{\leq
      j}\Sfnt{C}_{\leq j})^{\beta}\notag\\
    &= T_{i}\mu_e(\Sfnt{C}_{\leq i}|\Sfnt{Z}_{\leq i})^{\beta},
    \label{eq:testmart_proof}
  \end{align}
  with $T_{i}=\prod_{j=1}^{i}F_{j}$. To obtain the last equality above, we used Eq.~\eqref{eq:prob_de_compos}.  
  Thus, for each $e$, the sequence $Q_0 = 1$ and $Q_i=T_{i}\mu_e(\Sfnt{C}_{\leq i}|\Sfnt{Z}_{\leq i})^{\beta}$ for $i>0$
  satisfies the supermartingale property $\mathbb E_{\mu_e} (Q_{i+1}|\Sfnt{C}_{\leq i}\Sfnt{Z}_{\leq
      i})\le Q_i$. We remark that as a consequence, 
      $\Exp_{\mu_e}(Q_{i+1})=\Exp_{\mu_e}\big(\Exp_{\mu_e}(Q_{i+1}|\Sfnt{C}_{\leq
  i}\Sfnt{Z}_{\leq i})\big) \leq \Exp_{\mu_e}(Q_{i})$.  By induction
  this gives $\Exp_{\mu_e}(Q_{n})=\Exp_{\mu_e}(T_{n} \mu_e(\Sfnt{C}|\Sfnt{Z})^{\beta}) \leq 1$.  
  Thus, considering that $T_{n}=\prod_{i=1}^{n}F_{i}\geq 0$, $T_{n}$ is a PEF 
  with power $\beta$ for $\cH(\cC)$, that is, chaining PEFs yields PEFs for chained models.
  
  In Lem.~\ref{lem:confidence_interval}, if we replace $T_i$ and $U_{i}$ there
  by $T_i$ and $\mu_e(\Sfnt{C}_{\leq i}|\Sfnt{Z}_{\leq i})^{-\beta}$ here, then 
  from Eq.~\eqref{eq:itestmart_markov} and manipulating the inequality
  inside $\Prob(.)$, we get the inequality in Eq.~\eqref{eq:thm:uest_constr:pefupe}.
\end{proof}

That $F_{i+1}$ can be parameterized in terms of the past as
$F_{i+1}=F_{i+1}(C_{i+1}Z_{i+1}; \Sfnt{C}_{\leq i}\Sfnt{Z_{\leq i}})$
allows for adapting the PEFs based on $\Sfnt{C}\Sfnt{Z}$, but
no other information can be used.  To adapt the PEF $F_{i+1}$ 
based on other past information besides $\Sfnt{C}_{\leq i}\Sfnt{Z_{\leq i}}$, 
we need a ``soft'' generalization of probability estimation as detailed in Ref.~\cite{knill:2017}.

\subsection{Smooth Min-Entropy by Probability Estimation: Proof of Main Text Thm.~\ref{thm:smooth_min_entropy_bound}}
\label{sec:minentropy_extraction}

We want to generate bits that are near-uniform
conditional on $E$ and often other variables such as $\Sfnt{Z}$. For
our analyses, $E$ is not particularly an issue because our results
hold uniformly for all values of $E$, that is, conditionally on
$\{E=e\}$ for each $e$. However this is not the case for
$\Sfnt{Z}$. For this subsection, it is not necessary to
  structure the RVs as stochastic sequences, so below we use $C$ and
  $Z$ in place of $\Sfnt{C}$ and $\Sfnt{Z}$. 

\begin{definition}\label{def:smoothminentropy}
  The distribution $\mu$ of $CZE$ has \emph{$\epsilon$-smooth average
    $ZE$-conditional maximum probability $p$} if there exists a
  distribution $\nu$ of $CZE$ with $\TV(\nu,\mu)\leq\epsilon$ and
  $\sum_{ze}\max_{c}(\nu(c|ze))\nu(ze) \leq p$. The minimum $p$ for
  which $\mu$ has $\epsilon$-smooth average $ZE$-conditional
  maximum probability $p$ is denoted by $P^{\epsilon}_{\max,\mu}(C|ZE)$.  The quantity
  $H^{\epsilon}_{\min,\mu}(C|ZE)=-\log_{2}(P^{\epsilon}_{\max,\mu}(C|ZE))$ is 
  the \emph{ (classical) $\epsilon$-smooth $ZE$-conditional min-entropy}.
\end{definition}

We denote the $\epsilon$-smooth $ZE$-conditional min-entropy evaluated
  conditional on an event $\{\phi\}$ by $H_{\min}^{\epsilon}(C|ZE;\phi)$. 
We refer to the smoothness parameters as ``error bounds''.  Observe
that the definitions are monotonic in the error bound.  For example,
if $P^{\epsilon}_{\max,\mu}\leq p$ and $\epsilon'\geq \epsilon$, then
$P^{\epsilon'}_{\max,\mu}\leq p$.  The quantity
$\sum_{ze}\max_{c}(\nu(c|ze))\nu(ze)$ in the definition of
$P^{\epsilon}_{\max,\mu}$ can be recognized as the (average) maximum
guessing probability of $C$ given $Z$ and $E$ (with respect to $\nu$),
whose negative logarithm is the guessing entropy defined, for example,
in Ref.~\cite{konig:2008}.

A summary of the relationships between smooth conditional min-entropies 
and randomness extraction with respect to quantum side information
is given in Ref.~\cite{koenig:qc2009a} and can be specialized to classical 
side information. When so specialized, the definition of the smooth 
conditional min-entropy in, for example, Ref.~\cite{koenig:qc2009a} 
differs from the one above in that Ref.~\cite{koenig:qc2009a} uses one 
of the fidelity-related distances. One such distance reduces to the 
Hellinger distance $h$ for probability distributions for which 
$h^{2}\leq \TV \leq \sqrt{2}h$.

The $Z$-conditional maximum probabilities with respect to $E=e$ can be lifted to 
the $ZE$-conditional maximum probabilities, as formalized by the next lemma.

\begin{lemma}\label{lm:pmax_noe}
  Suppose that for all $e$, $P^{\epsilon_{e}}_{\max,\mu(CZ|e)}(C|Z)\leq p_{e}$, 
  and let $\bar\epsilon=\sum_{e}\epsilon_{e}\mu(e)$ and $\bar p=\sum_{e}p_{e}\mu(e)$. 
  Then $P^{\bar\epsilon}_{\max,\mu(CZE)}(C|ZE)\leq \bar p$.
\end{lemma}

\begin{proof}
  For each $e$, let $\nu_{e}$ witness $P^{\epsilon_{e}}_{\max,\mu(CZ|e)}(C|Z)\leq p_{e}$.
  Then $\TV(\nu_{e},\mu(CZ|e))\leq\epsilon_{e}$ and $\sum_{z}\max_{c}(\nu_{e}(c|z))\nu_{e}(z) \leq p_{e}$.
  Define $\nu$ by $\nu(cze)=\nu_{e}(cz)\mu(e)$. Then the marginals $\nu(E)=\mu(E)$, so we can apply 
  Eq.~\eqref{eq:tv_samemarg} for 
  \begin{equation*}
    \TV(\nu,\mu)=\sum_{e}\TV(\nu_{e},\mu(CZ|e))\mu(e)
    \leq \sum_{e}\epsilon_{e}\mu(e)=\bar\epsilon.
  \end{equation*}
  Furthermore, 
    \begin{align}
    \sum_{ze}\max_{c}(\nu(c|ze))\nu(ze)
    &= \sum_{e}\mu(e)\sum_{z}\max_{c}(\nu_{e}(c|z))\nu_e(z)\notag\\
    &\leq \sum_{e}\mu(e) p_{e} = \bar p, \notag
  \end{align}
  as required for the conclusion.
\end{proof}

The level of a probability estimator relates to the smoothness
parameter for smooth min-entropy via the relationships 
established below.  

\begin{theorem}\label{thm:esupe_fail_esmaxprob}
  Suppose that $U$ is an $\epsilon$-$\UPE(C|ZE; \cH)$ 
  and that the distribution $\mu$ of $CZE$ satisfies the model $\cH$.  
  Let $p\geq 1/|\Rng(C)|$ and $\kappa=\mu(U\leq p)$. 
  Then $P^{\epsilon/\kappa}_{\max,\mu(CZE|U\leq p)}(C|ZE) \leq
  p/\kappa$.
\end{theorem}

\begin{proof}
  Let $\kappa_{e}=\mu(U\leq p|e)$.
  Below we show that for all values $e$ of $E$,
  $P^{\epsilon/\kappa_{e}}_{\max,\mu(CZ|e,U\leq p)}(C|Z)\leq p/\kappa_{e}$.
  Once this is shown, we can use
  \begin{equation}
    \sum_{e}\frac{1}{\kappa_{e}}\mu(e|U\leq p) 
    = \sum_{e}\frac{1}{\mu(U\leq p|e)}\mu(e|U\leq p)
    = \sum_{e}\frac{\mu(e)}{\mu(U\leq p)} = 1/\kappa,
  \end{equation}
  and Lem.~\ref{lm:pmax_noe} to complete the proof.  For the
  remainder of the proof, $e$ is fixed, so we simplify the notation
  by universally conditioning on $\{E=e\}$ and omitting the explicit
  condition. Further, we omit $e$ from suffixes.  Thus
  $\kappa=\kappa_{e}$ from here on.  

  Let $\kappa_{z}=\mu(U\leq p|z)$.  We have $\sum_{z}\kappa_{z}\mu(z)=\kappa$ and 
  \begin{equation} 
  \kappa_{z}=\mu(z|U\leq p)\kappa/\mu(z). 
  \label{eq:kappa}
  \end{equation}
  Define the subprobability distribution $\tilde\mu$ by
  $\tilde\mu(cz)=\mu(cz)\knuth{U(cz)\geq \mu(c|z)}$. By the definition
  of $\epsilon$-UPEs, we get that the weight of $\tilde{\mu}$ satisfies 
  \begin{align}\label{eq:weight_mu}
  w(\tilde{\mu})&=\sum_{cz}\mu(cz)\knuth{U(cz)\geq\mu(c|z)}\notag\\
    &= \mu(U(CZ)\geq\mu(C|Z))\notag\\
    &\geq 1-\epsilon.
  \end{align}

  Define $\tilde\nu(cz) = \tilde\mu(cz)\knuth{U(cz)\leq p}/\kappa$.
  The weight of $\tilde\nu$ satisfies
  \begin{align} 
    w(\tilde\nu) &= \sum_{cz}\tilde\mu(cz)\knuth{U(cz)\leq p}/\kappa\notag\\
    &\leq \sum_{cz}\mu(cz)\knuth{U(cz)\leq p}/\kappa\notag\\
    & = \mu(U\leq p)/\kappa = 1,\\ 
    w(\tilde\nu) &= 
    \sum_{cz}\mu(cz)\knuth{U(cz)\leq p}/\kappa 
    -\sum_{cz}(\mu(cz)-\tilde\mu(cz))\knuth{U(cz)\leq p}/\kappa \notag\\
    &=1-\sum_{cz}(\mu(cz)-\tilde\mu(cz))\knuth{U(cz)\leq p}/\kappa \notag\\
    &\geq 1-\sum_{cz}(\mu(cz)-\tilde\mu(cz))/\kappa = 1- (1-w(\tilde\mu))/\kappa \notag\\
    &\geq 1- (1-(1-\epsilon))/\kappa = 1-\epsilon/\kappa.
    \label{eq:weight_nu2}
  \end{align}
  To obtain the last inequality above, we used Eq.~\eqref{eq:weight_mu}.
  Thus $\tilde\nu$ is a subprobability distribution of weight at least
  $1-\epsilon/\kappa$.  We use $\tilde\nu$ to construct the distribution $\nu$
  witnessing the conclusion of the theorem. For each $cz$ we bound
  \begin{align}
    \tilde\nu(cz)/\mu(z|U\leq p) 
    &= \mu(cz)\knuth{U(cz)\geq \mu(c|z)}\knuth{U(cz)\leq p}
    /(\kappa \mu(z|U\leq p))
    \notag\\
    &= \mu(c|z)\knuth{U(cz)\geq \mu(c|z)}\knuth{U(cz)\leq p}/\kappa_{z}
    \notag\\
    &\leq p /\kappa_{z},\label{eq:lm:esupe_fail_esmaxprob1}
  \end{align}
  where in the second step we used Eq.~\eqref{eq:kappa}. 
  Define $\tilde\nu(C|z)$ by $\tilde\nu(c|z)=\tilde\nu(cz)/\mu(z|U\leq p)$, 
  with $\tilde\nu(c|z)=0$ if $\mu(z|U\leq p)=0$, and 
  let $w_{z}=w(\tilde\nu(C|z))$. We show below that $w_{z}\leq 1$, 
  and so the definition of $\tilde\nu(C|z)$ extends the conditional 
  probability notation to the
  subprobability distribution $\tilde{\nu}$ with the understanding
  that the conditionals are with respect to $\mu$ given $\{U\leq
  p\}$. Applying the first two steps of
  Eq.~\eqref{eq:lm:esupe_fail_esmaxprob1} and continuing from there, we
  have
  \begin{align}\label{eq:weight_conditional_nu}
    \tilde\nu(c|z) &= 
    \mu(c|z)\knuth{U(cz)\geq \mu(c|z)}\knuth{U(cz)\leq p}/\kappa_{z}\notag\\
    &\leq \mu(c|z)\knuth{U(cz)\leq p}/\kappa_{z}\notag\\
    &=\mu(c,U\leq p|z)/\mu(U\leq p|z) = \mu(c|z,U\leq p).
  \end{align}
  Since $\mu(C|z,U\leq p)$ is a normalized distribution, 
  the above equation implies that $w_{z}\leq 1$. 
  For each $z$, we have that $\tilde\nu(C|z)\leq p/\kappa_z$ 
  (Eq.~\eqref{eq:lm:esupe_fail_esmaxprob1}), $p/\kappa_{z}\geq p\geq 1/|\Rng(C)|$,  
  and $\mu(C|z,U\leq p)$ dominates $\tilde\nu(C|z)$ (Eq.~\eqref{eq:weight_conditional_nu}). 
  Hence, we can apply Lem.~\ref{lm:tvsub} to obtain
  distributions $\nu_{z}$ of $C$ such that $\nu_{z}\geq
  \tilde\nu(C|z)$, $\nu_{z}\leq p/\kappa_{z}$, and
  $\TV(\nu_{z},\mu(C|z,U\leq p))\leq 1-w_{z}$.  Now we can define
  the distribution $\nu$ of $CZ$ by $\nu(cz)=\nu_{z}(c)\mu(z|U\leq p)$.  
  By Eq.~\eqref{eq:tv_samemarg}, we get
  \begin{align}
    \TV(\nu,\mu(CZ|U\leq p)) &= \sum_{z}\TV(\nu_{z},\mu(C|z,U\leq p))\mu(z|U\leq p)\notag\\
    &\leq \sum_{z}(1-w_{z})\mu(z|U\leq p)\notag\\
    &= 1-\sum_{z}w(\tilde\nu(C|z))\mu(z|U\leq p)\notag\\
    &= 1-\sum_{z}\sum_{c}(\tilde\nu(cz)/\mu(z|U\leq p))\mu(z|U\leq p)\notag\\
    &= 1-w(\tilde\nu)\leq \epsilon/\kappa,\label{eq:lm:esupe_fail_esmaxprob_1}
  \end{align}
  where in the last step we used Eq.~\eqref{eq:weight_nu2}.
  For the average maximum probability of $\nu$, we get
  \begin{align}
    \sum_{z}\max_{c}\nu(c|z)\nu(z) &=
    \sum_{z}\max_{c}\nu_{z}(c)\mu(z|U\leq p)\notag\\
    &\leq p\sum_{z}\mu(z|U\leq p)/\kappa_{z}\notag\\
    &= p\sum_{z}\mu(z)/\kappa = p/\kappa,
  \end{align}
  where to obtain the last line we used Eq.~\eqref{eq:kappa}. The above
  two equations show that for an arbitrary value $e$ of $E$,
  $P^{\epsilon/\kappa_{e}}_{\max,\mu(CZ|e,U\leq p)}(C|Z)\leq p/\kappa_{e}$,
  which together with the argument at the beginning of the proof
  establishes the theorem.
\end{proof}

The above theorem implies Thm.~\ref{thm:smooth_min_entropy_bound} in the main text as a corollary.  
\begin{corollary} \label{cor:smooth_min_entropy_bound} 
Suppose that the distribution $\mu$ of $CZE$ satisfies the chained model $\cH(\cC)$. 
Let $1\geq p\geq 1/|\Rng(C)|$ and $1\geq \kappa', \epsilon>0$. 
  Define $\{\phi\}$ to be the event that $U\leq p$, where $U$ is given in 
  Eq.~\eqref{eq:upe_final}. 
  Let $\kappa'\leq\kappa=\Prob_{\mu}(\phi)$.  Then the smooth conditional 
  min-entropy satisfies
  \begin{equation*}
    H_{\min}^{\epsilon}(C|ZE;\phi) \geq -\log_2(p/\kappa'^{1+1/\beta}). 
  \end{equation*}
\end{corollary}

\begin{proof}
We observe that the event that $U\leq p$ is the same as the event that  
$U'\leq p/\kappa^{1/\beta}$, where $U'=(T_{n}\epsilon\kappa)^{-1/\beta}$
and $T_n$ is defined as above Eq.~\eqref{eq:upe_final}. By 
Thm.~\ref{thm:uest_constr}, $U'$ is an $\epsilon\kappa$-$\UPE$. 
In Thm.~\ref{thm:esupe_fail_esmaxprob}, if we replace $U$ and $p$ there by 
$U'$ and $p/\kappa^{1/\beta}$ here, then we obtain 
$P^{\epsilon}_{\max,\mu(CZE|\phi)}(C|ZE)\leq p/\kappa^{1+1/\beta}$. Since
$\kappa'\leq\kappa$, 
we also have $P^{\epsilon}_{\max,\mu(CZE|\phi)}(C|ZE)\leq p/\kappa'^{1+1/\beta}$.
According to the definition of the smooth conditional min-entropy in 
Def.~\ref{def:smoothminentropy}, we get the lower bound in the corollary. 
\end{proof}

We remark that, to obtain uniformly random bits, Cor.~\ref{cor:smooth_min_entropy_bound} 
can be composed directly with ``classical-proof'' strong extractors in a complete protocol  
for randomness generation. The error bounds from the corollary and those of
the extractor compose additively~\cite{knill:2017}. Efficient randomness extractors requiring 
few seed bits exist, see Refs.~\cite{trevisan:2001, mauerer:2012}. Specific instructions for 
ways to apply them for randomness generation can be found in Refs.~\cite{bierhorst:qc2017a,
bierhorst:qc2018a, knill:2017}.

\section{Properties of PEFs} 
\label{subsec:pef_property}

Here we prove the monotonicity of the functions $g(\beta)$ and $\beta g(\beta)$: 
As $\beta$ increases, the rate $g(\beta)$ as defined in Eq.~\eqref{eq:gain_rate}
of the main text is monotonically non-increasing, and $\beta g(\beta)$
is monotonically non-decreasing. These are the consequence of the following lemma:

\begin{lemma}
  \label{lm:power_change}
  If $F$ is a PEF with power $\beta$ for the trial model $\cC$, then for any $0<\gamma\leq 1$, 
  $F$ is a PEF with power $\beta/\gamma$ for $\cC$, and $F^{\gamma}$ is a PEF with power 
  $\gamma\beta$ for $\cC$.  
\end{lemma}

\begin{proof}

  For an arbitrary distribution $\sigma\in\cC$, we have $0\leq \sigma(c|z)\leq 1$ for all $cz$. 
  By the monotonic property of the exponential function $x\mapsto a^{x}$ with $0\leq a\leq 1$, 
  we get that $\sigma(c|z)^{\beta/\gamma}\leq \sigma(c|z)^{\beta}$ for all $cz$. Therefore,
  if a non-negative RV $F$ satisfies that 
  \begin{equation*}
    \sum_{cz}F(cz)\sigma(c|z)^{\beta}\sigma(cz) \leq 1, 
  \end{equation*}
  then 
  \begin{equation*}
  \sum_{cz}F(cz)\sigma(c|z)^{\beta/\gamma}\sigma(cz)\leq \sum_{cz}F(cz)\sigma(c|z)^{\beta}\sigma(cz) \leq 1.
  \end{equation*}
  Hence, if $F$ is a PEF with power $\beta$ for $\cC$, then $F$ 
  is a PEF with power $\beta/\gamma$ for $\cC$. 
  
  On the other hand, by the concavity of the function $x\mapsto x^{\gamma}$ with $0<\gamma\leq 1$, we can 
  apply Jensen's inequality to get
  \begin{align}
    \Exp_{\sigma}\left(F(CZ)^{\gamma}\sigma(C|Z)^{\gamma\beta}\right)
    & = \Exp_{\sigma}\left(\left(F(CZ)\sigma(C|Z)^{\beta}\right)^{\gamma}\right) \notag \\
    & \leq \left(\Exp_{\sigma}\left(F(CZ)\sigma(C|Z)^{\beta}\right)\right)^{\gamma} \notag \\
    & \leq 1, \notag
  \end{align}
  for all distributions $\sigma\in\cC$. Hence $F^{\gamma}$ is a PEF with 
  power $\gamma\beta$ for $\cC$.  
\end{proof}

The property that $\beta g(\beta)$ is monotonically non-decreasing in $\beta$ follows 
directly from Lem.~\ref{lm:power_change} and the definition of $g(\beta)$ in Eq.~\eqref{eq:gain_rate}
of the main text. On the other hand, to prove that $g(\beta)$ is monotonically non-increasing in 
$\beta$, we also need to use the equality that  
\begin{equation*}
 \Exp_{\sigma} \big(\log_2(F^{\gamma}(CZ))/(\gamma\beta)\big)=\Exp_{\sigma} \big(\log_2(F(CZ))/\beta\big).
\end{equation*}

The monotonicity of the function $g(\beta)$ (or $\beta g(\beta)$)
helps to determine the maximum asymptotic randomness rate 
$g_0=\sup_{\beta>0}g(\beta)$ (or the maximum certificate rate 
$\gamma_{\textrm{PEF}}=\sup_{\beta>0}\beta g(\beta)$), as one can
analyze the PEFs with powers $\beta$ only in the limit where $\beta$
goes to $0$ (or where $\beta$ goes to the infinity).

\section{Numerical Optimization of PEFs}
\label{subsec:optimization}

We provide more details here on how to perform the optimizations (such
as the optimization in Eq.~\eqref{eq:opt_pef} of the main text) required to determine
the power $\beta$ and the PEFs $F_i$ to be used at the $i$'th trial.  
We claim that to verify that the PEF $F$ satisfies the first constraint in Eq.~\eqref{eq:opt_pef}  
 of the main text for all $\sigma\in\cC$, it suffices to check this constraint 
 on the extremal members of the convex closure of $\cC$. The claim follows from the 
 next lemma, Carath\'eodory's theorem, and induction on the number of terms in 
 a finite convex combination.
\begin{lemma}\label{lm:convexindist}
  Let $F\geq 0$ and $\beta>0$.  Suppose that the distribution $\sigma$ can be expressed 
  as a convex combination of two distributions: For all $cz$, $\sigma(cz) =
  \lambda\sigma_1(cz)+(1-\lambda)\sigma_2(cz)$ with $\lambda \in [0,1]$. 
  If the distributions $\sigma_1$ and $\sigma_2$ satisfy 
  $\sum_{cz}F(cz)\sigma_i(c|z)^\beta\sigma_i(cz)\le 1$, 
  then $\sigma$ satisfies $\sum_{cz}F(cz)\sigma(c|z)^\beta\sigma(cz)\le 1$.
\end{lemma}

\begin{proof}
We start by proving that for every $cz$, the following inequality holds:
\begin{equation}\label{eq:convbeta}
\sigma(c|z)^\beta\sigma(cz) \le \lambda\sigma_1(c|z)^\beta\sigma_1(cz)+(1-\lambda)\sigma_2(c|z)^\beta\sigma_2(cz).
\end{equation}
If $\sigma_1(z)=\sigma_2(z)=0$, we recall our convention that
probabilities conditional on $z$ are zero, and so for every $c$,
$\sigma_1(c|z)=\sigma_2(c|z)=\sigma(c|z)=0$. Hence,
Eq.~\eqref{eq:convbeta} holds immediately (as an equality). If
$\sigma_1(z)=0<\sigma_2(z)$, then for every $c$, $\sigma_1(c|z)=0$ and
$\sigma(cz)=(1-\lambda)\sigma_2(cz)$. In this case, one can verify
that Eq.~\eqref{eq:convbeta} holds. By symmetry, 
Eq.~\eqref{eq:convbeta} also holds in the case that 
$\sigma_2(z)=0<\sigma_1(z)$.  Now consider the case that 
$\sigma_1(z)>0$ and $\sigma_2(z)>0$. Let $x_i=\sigma_i(cz)$ and 
$y_i = \sigma_i(z)$, and consider the function
\begin{equation*}
f(\lambda) = (\lambda x_1 + (1-\lambda) x_2)^{1+\beta}(\lambda y_1 +(1-\lambda)y_2)^{-\beta},
\end{equation*}
so $f(0)=\sigma_2(c|z)^\beta\sigma_2(cz)$, $f(1)=\sigma_1(c|z)^\beta\sigma_1(cz)$,  and $f(\lambda) = \sigma(c|z)^\beta\sigma(cz)$. If we can show that $f(\lambda)$ is convex in $\lambda$ on the interval $[0,1]$, Eq.~\eqref{eq:convbeta}  will follow. Since $f(\lambda)$ is continuous for $\lambda\in[0,1]$ and smooth for $\lambda\in(0,1)$, it suffices to 
show that $f''(\lambda)\geq 0$ as follows:
  \begin{align*}
    f'(\lambda) &= (\lambda x_1+(1-\lambda)x_{2})^{\beta}(\lambda y_1+(1-\lambda)y_2)^{-\beta-1}\notag\\
    &\hphantom{=\;\;} \times \Big( (1+\beta) (x_{1}-x_{2})(\lambda y_1+(1-\lambda)y_{2}) + 
    (-\beta) (\lambda x_{1}+(1-\lambda)x_{2})(y_1-y_{2}) \Big)\label{eq:thm:prob_constr_2}\\
    f''(\lambda) &= 
    (\lambda x_{1}+(1-\lambda)x_{2})^{\beta-1}(\lambda y_1+(1-\lambda)y_{2})^{-\beta-2}\notag\\
    &\hphantom{=\;\;} \times \Big( \beta(1+\beta) (x_{1}-x_{2})^{2}(\lambda y_1+(1-\lambda)y_{2})^{2} 
    \notag\\
    &\hphantom{=\;\;} \hphantom{\times\Big(} + 2(-\beta)(1+\beta) (x_{1}-x_{2})(y_1-y_{2})(\lambda x_{1}+(1-\lambda) x_{2})(\lambda y_1+(1-\lambda)y_{2})\notag\\
    &\hphantom{=\;\;} \hphantom{\times\Big(} + (-\beta)(-1-\beta) (y_1-y_{2})^{2}(\lambda x_{1}+(1-\lambda )x_{2})^{2}\Big)\notag\\
    &= (\lambda x_{1}+(1-\lambda)x_{2})^{\beta-1}(\lambda y_1+(1-\lambda)y_{2})^{-\beta-2}\notag\\
    &\hphantom{=\;\;} \times \beta(1+\beta)\Big((x_{1}-x_{2})(\lambda y_1+(1-\lambda) y_{2}) -
    (y_1-y_{2})(\lambda x_{1}+(1-\lambda) x_{2})\Big)^{2},
  \end{align*}
  which is a non-negative multiple of a square. Having demonstrated Eq.~\eqref{eq:convbeta}, we can complete the proof of the lemma as follows:
 \begin{eqnarray*}
 \sum_{cz}F(cz)\sigma(c|z)^\beta\sigma(cz) &\le& \sum_{cz}F(cz)\left[\lambda\sigma_1(c|z)^\beta\sigma_1(cz)+(1-\lambda)\sigma_2(c|z)^\beta\sigma_2(cz)\right]\\
 &=&\lambda \sum_{cz}F(cz)\sigma_1(c|z)^\beta\sigma_1(cz)+(1-\lambda)\sum_{cz}F(cz)\sigma_2(c|z)^\beta\sigma_2(cz)\\
 &\le& \lambda \times 1 + (1-\lambda)\times 1\\
 &=&1.
 \end{eqnarray*}
\end{proof}

 Suppose that the trial model $\cC$ is a convex polytope with a finite number of extremal distributions 
 $\sigma_k(CZ)$, $k=1,2,...,K$. In view of the claim before Lem.~\ref{lm:convexindist}, the 
 optimization problem in Eq.~\eqref{eq:opt_pef} of the main text is equivalent to %
\begin{equation}
  \begin{array}[b]{ll}
    \textrm{Max:}& n\Exp_{\nu} \log_2(F(CZ))/\beta+\log_2(\epsilon)/\beta \\ 
    \textrm{With:}& \sum_{cz}F(cz)\sigma_k(c|z)^{\beta}\sigma_k(cz) \leq 1, k=1,2,...,K, \\
    & F(cz)\geq 0, \forall cz.
  \end{array}\label{eq:opt_pef_simplified}
\end{equation}
Given the values of $n$, $\beta$, $\epsilon$, $\nu$, and $\sigma_k$ with $k=1,2,...,K$, 
the objective function in Eq.~\eqref{eq:opt_pef_simplified} is a concave function of $F(CZ)$, and 
each constraint on $F(CZ)$ is linear. Hence, the above optimization problem can be solved by any 
algorithm capable of optimizing nonlinear functions with linear constraints on the arguments. 
In our implementation, we use sequential quadratic programming. Due to numerical imprecision, it 
is possible that the returned numerical solution does not satisfy the first constraint in
Eq.~\eqref{eq:opt_pef_simplified} and the corresponding PEF is not valid. In this case, we can multiply 
the returned numerical solution by a positive factor smaller than 1, whose value is given by the reciprocal
of the largest left-hand side of the above first constraint at the extremal distributions $\sigma_k(CZ)$, 
$k=1,2,...,K$.  Then, the re-scaled solution is a valid PEF. We remark that if the trial model $\cC$ is
not a convex polytope but there exists a good approximation $\cC\subseteq\cD$ with 
$\cD$ a convex polytope, then we can enlarge the model to $\cD$ for an effective method to 
determine good PEFs.  

Consider device-independent randomness generation (DIRG) in the CHSH
Bell-test configuration~\cite{clauser:qc1969a} with inputs $Z=XY$
and outputs $C=AB$, where $A, B, X, Y\in\{0,1\}$. If the input
distribution $\Prob(XY)$ is fixed with $\Prob(xy)>0$ for all $xy$, then we 
need to characterize the set of input-conditional output distributions
$\Prob(AB|XY)$. If we consider all distributions $\Prob(AB|XY)$
satisfying non-signaling conditions~\cite{PRBox}, then the associated trial model 
$\cC$ is the non-signaling polytope, which is convex and has $24$ extreme
points~\cite{barrett:2005}.  If we consider only the distributions
$\Prob(AB|XY)$ achievable by quantum mechanics, then the associated 
trial model is a proper convex subset of the above non-signaling polytope.  
The quantum set has an infinite number of extreme points.
In our analysis of the Bell-test results reported in Refs.~\cite{pironio:2010, shalm:2015}, 
we simplified the problem by considering instead the set of distributions 
$\Prob(AB|XY)$ satisfying non-signaling conditions~\cite{PRBox} and Tsirelson's bounds~\cite{Tsirelson:1980}, 
which includes all the distributions $\Prob(AB|XY)$ achievable by quantum mechanics. For a fixed input 
distribution $\Prob(XY)$ with $\Prob(xy)>0$ for all $xy$, the associated trial model $\cC$ is a convex polytope 
with $80$ extreme points~\cite{knill:2017}.  If the input distribution $\Prob(XY)$ is not fixed but is 
contained in a convex polytope, the associated trial model $\cC$ is still a convex polytope (see Ref.~\cite{knill:2017} for more details).  Therefore, for DIRG based on the CHSH Bell test~\cite{clauser:qc1969a}, the optimizations 
for determining the power $\beta$ and the PEFs $F_i$ can be expressed in the form in 
Eq.~\eqref{eq:opt_pef_simplified} and hence solved effectively.

\section{Relationship between Certificate Rate and Statistical Strength} 
\label{subsec:certificate_rate}

We prove that for DIRG in the CHSH Bell-test configuration, the maximum certificate rate
$\gamma_{\textrm{PEF}}$ witnessed by PEFs at a distribution $\nu$ of trial results is
equal to the statistical strength of $\nu$ for rejecting local realism
as studied in Refs.~\cite{vanDam:2005, acin2005, zhang:2010}.  To prove this, we
first simplify the optimization problem for determining
$\gamma_{\textrm{PEF}}$. Then, we show that the simplified
optimization problem is the same as that for determining the
statistical strength. The argument generalizes to any convex-polytope
  model whose extreme points are divided into the following 
  two classes: 1) classical deterministic
  distributions satisfying that given the inputs, the outputs are
  deterministic (here we require that for every $cz$ there exists 
  a distribution in the model where the outcome is $c$ given $z$), 
  and 2) distributions that are
  completely non-deterministic in the sense that for no input is the
  output deterministic. The argument further generalizes to models 
  contained in such a model, provided it includes all of the classical
  deterministic distributions of the outer model. 

In order to determine $\gamma_{\textrm{PEF}}=\sup_{\beta>0} \beta g(\beta)$, 
considering the monotonicity of the function $\beta g(\beta)$ proved 
in Sect.~\ref{subsec:pef_property} and the definition of $ g(\beta)$ 
in Eq.~\eqref{eq:gain_rate} of the main text, we need to solve the following 
optimization problem at arbitrarily large powers $\beta$:
\begin{equation}
  \begin{array}[b]{ll}
    \textrm{Max:}& \Exp_{\nu} \big(\log_2(F(CZ))\big) \\ 
    \textrm{With:}& \sum_{cz}F(cz)\sigma(c|z)^{\beta}\sigma(cz) \leq 1
    \textrm{\ for all $\sigma\in\cC$}, \\
    & F(cz)\geq 0, \textrm{\ for all $cz$}.
  \end{array}\label{eq:gamma_opt1}
\end{equation}
To simplify this optimization, we first consider the case that the
trial model $\cC$ is the set of non-signaling distributions with a 
fixed input distribution $\Prob(Z)$ where $\Prob(z)>0$ for all $z$.  
The model $\cC$ is a convex polytope and has $24$ 
extremal distributions~\cite{barrett:2005}, among which there
are $16$ deterministic local realistic distributions, denoted by
$\sigma_{\text{LR}_i}$, $i=1,2,...,16$, and $8$ variations of the
Popescu-Rohrlich (PR) box~\cite{PRBox}, denoted by
$\sigma_{\text{PR}_j}$, $j=1,2,...,8$. According to the discussion in
Sect.~\ref{subsec:optimization}, the optimization problem in
Eq.~\eqref{eq:gamma_opt1} is equivalent to
\begin{equation}
  \begin{array}[b]{ll}
    \textrm{Max:}& \Exp_{\nu} \big(\log_2(F(CZ))\big) \\ 
    \textrm{With:}& \sum_{cz}F(cz)\sigma_{\text{LR}_i}(cz) \leq 1, \forall i,\\ 
     & \sum_{cz}F(cz)\sigma_{\text{PR}_j}(c|z)^{\beta}\sigma_{\text{PR}_j}(cz) \leq 1, \forall j,\\
     & F(cz)\geq 0, \textrm{\ for all $cz$},
  \end{array}\label{eq:gamma_opt2}
\end{equation}
where we used the fact that $\sigma_{\text{LR}_i}(c|z)$ is either $0$ or $1$. 
Only the second constraint in Eq.~\eqref{eq:gamma_opt2} depends on the power $\beta$. 
The distributions $\sigma_{\text{PR}_{j}}$ satisfy that $\sigma_{\text{PR}_{j}}(c|z)<1$ 
for all $cz$. Hence $\sigma_{\text{PR}_j}(c|z)^{\beta}\rightarrow 0$ for all $cz$ as 
$\beta\rightarrow\infty$. Because there are finitely many constraints and values of 
$cz$, the second constraint becomes irrelevant for sufficiently large $\beta$.  
  Let $\beta_{\textrm{th}}^{\textrm{NS}}$ be the minimum $\beta$ for
  which the second constraint is implied by the first.  The threshold 
  $\beta_{\textrm{th}}^{\textrm{NS}}$ is independent of the specific 
  input distribution. To see this, the last factors in the sums on 
  the left-hand sides of the constraints in Eq.~\eqref{eq:gamma_opt2}
  are of the form $\sigma(cz)$, which can be written as $\sigma(c|z)\sigma(z)$ 
  with a fixed $\sigma(z)$.    We can define 
  $\tilde F(cz)=F(cz)\sigma(z)$ and optimize over $\tilde F$ instead, 
  thus eliminating the fixed input distribution from the problem.
  Then the first constraint on $\tilde F$ implies that 
  $\sum_{cz}\tilde F(cz)\sum_{i}\sigma_{\text{LR}_i}(c|z)\leq 16$.  
  Since $\sum_{i}\sigma_{\text{LR}_{i}}(c|z) \geq 4$ for each $cz$, this
  constraint implies the second provided that 
  $\sigma_{\text{PR}_j}(c|z)^{1+\beta}\leq 1/4$,
  which holds for each $j$ and $cz$ for sufficiently large $\beta$. 
  Particularly, since $\sigma_{\text{PR}_j}(c|z)$ is either $0$ or 
  $1/2$~\cite{barrett:2005}, we obtain that $\beta_{\textrm{th}}^{\textrm{NS}}\leq 1$. 
  Furthermore, by numerical optimization for a sample of large-enough $\beta$ 
  we find that $\beta_{\textrm{th}}^{\textrm{NS}}\approx 0.4151$.
  Therefore, when $\beta\geq\beta_{\textrm{th}}^{\textrm{NS}}$ the 
  optimization problem in Eq.~\eqref{eq:gamma_opt2} is independent 
  of $\beta$ and becomes
\begin{equation}
  \begin{array}[b]{ll}
    \textrm{Max:}& \Exp_{\nu} \big(\log_2(F(CZ))\big) \\ 
    \textrm{With:}& \sum_{cz}F(cz)\sigma_{\text{LR}_i}(cz) \leq 1, \forall i,\\   
     & F(cz)\geq 0, \textrm{\ for all $cz$}.
  \end{array}\label{eq:gamma_opt3}
\end{equation}
This optimization problem is identical to the one for designing the optimal 
test factors for the hypothesis test of local realism~\cite{zhang:2011, zhang:2013,
  knill:qc2014a}. In Ref.~\cite{zhang:2011} it is proven that the optimal value of 
  the optimization problem in Eq.~\eqref{eq:gamma_opt3} is equal to the
  statistical strength for rejecting local realism~\cite{vanDam:2005, acin2005, zhang:2010}, 
  which is defined as
\begin{equation*}
s=\min_{\sigma_{\textrm{LR}}} D_{\textrm{KL}}(\nu|\sigma_{\textrm{LR}}). 
\end{equation*}
Here, $\sigma_{\textrm{LR}}$ is an arbitrary local realistic distribution
and $D_{\textrm{KL}}(\nu|\sigma_{\textrm{LR}})$ is the Kullback-Leibler
divergence from $\sigma_{\textrm{LR}}$ to
$\nu$~\cite{KL_divergence}. Therefore, when $\beta\geq
\beta_{\textrm{th}}^{\textrm{NS}}$ we have $\beta
g(\beta)=s$. Considering that the function $\beta g(\beta)$ is
monotonically non-decreasing in $\beta$, we have shown that
\begin{equation*}
\gamma_{\textrm{PEF}}=\sup_{\beta>0} \beta g(\beta)=s.
\end{equation*} 


Now we consider the case where the trial model $\cC$ is the set of 
quantum-achievable distributions with a fixed input distribution 
$\Prob(Z)$ where $\Prob(z)>0$ for all $z$. Since the set of 
quantum-achievable distributions is a proper subset of the non-signaling 
polytope, the constraints on $F(CZ)$ imposed by quantum-achievable 
distributions are a subset of the constraints imposed by non-signaling
distributions. Moreover, the set of quantum-achievable distributions
contains all local realistic distributions. Therefore, in the quantum
case, when $\beta\geq \beta_{\textrm{th}}^{\textrm{NS}}$, the constraints
on $F(CZ)$ are also implied by the constraints associated with the 
local realistic distributions.  Consequently the maximum certificate rate 
$\gamma_{\textrm{PEF}}$ is also equal to the statistical strength $s$. 
 We remark that as a consequence, if we set $\beta_{\textrm{th}}^{\textrm{QM}}$ 
 to be the threshold such that when $\beta\geq \beta_{\textrm{th}}^{\textrm{QM}}$ 
 all quantum constraints on $F(CZ)$ are implied by those imposed by the local 
 realistic distributions, then $\beta_{\textrm{th}}^{\textrm{QM}}\leq
  \beta_{\textrm{th}}^{\textrm{NS}}$.

We remark that $\beta_0=\inf\{\beta|\beta g(\beta)=s\}$ is typically strictly 
less than $\beta_{\textrm{th}}^{\textrm{NS}}$ and depends on both the distribution 
$\nu$ as well as the trial model $\cC$. One way to understand this behavior is as
follows: When $\beta<\beta_{\textrm{th}}^{\textrm{NS}}$, the second constraint in 
Eq.~\eqref{eq:gamma_opt2} is relevant; however, if $\beta$ is still large enough, 
it is possible that the constraint does not affect the optimal solution of the 
optimization problem~\eqref{eq:gamma_opt2}. By numerical optimization, we find that 
for the CHSH Bell-test configuration $\beta_0$ is typically less than $0.2$ 
when the trial model $\cC$ includes all non-signaling distributions with the uniform 
distribution for inputs.

\section{Analytic Expressions for Asymptotic Randomness Rates} 
\label{subsec:gain_rate}
In this section we derive the asymptotic randomness rates for the
  trial model consisting of non-signaling distributions according
to two different methods for DIRG protocol based on the CHSH Bell 
test~\cite{clauser:qc1969a}.  We first consider
the maximum asymptotic rate $g_0$ witnessed by PEFs. Then, we derive the
single-trial conditional min-entropy for comparison.  

Suppose that the distribution of each trial's inputs $XY$ and outputs
$AB$ is $\nu(ABXY)\in \cC$, where $\cC$ is the model for each
trial. The maximum asymptotic rate $g_0$ is equal to the worst-case
conditional entropy that is consistent with the distribution
$\nu(ABXY)$~\cite{knill:2017}. That is, the rate $g_0$ is given 
by the following minimization:
\begin{equation}
g_0=\min_{\sigma}\left\{H_{\sigma}(AB|XYE): \sigma(ABXY)=\nu(ABXY)\right\},
\label{eq:analytic_g1}
\end{equation}
where $\sigma$ is the joint distribution of $A, B, X, Y$ and $E$, and $\sigma(ABXY)$ 
is its marginal. By the assumption that the value space of $E$ is countable, 
we can also express the above minimization as  
\begin{equation}
g_0=\min_{\omega_e,\sigma_e} \big\{\sum_{e}\omega_{e}H_{\sigma_e}(AB|XY,E=e):\forall e, \sigma_{e}\in\cC\text{ and } \omega_{e}\geq 0,
\sum_{e}\omega_{e}=1, \sum_{e}\omega_{e}\sigma_{e}=\nu \big\},
\label{eq:analytic_g2}
\end{equation}
where $\sigma_e$ is the distribution of $A, B, X$ and $Y$ conditional on $E=e$ according to $\sigma$, and 
$\omega_e$ is the probability of the event $E=e$.  By the concavity of the conditional entropy, if any of 
the $\sigma_e$ contributing to the sum in Eq.~\eqref{eq:analytic_g2}
is not extremal in $\cC$, we can replace it by a convex combination of extremal distributions 
to decrease the value of the sum.  Thus, we only have to consider extremal distributions 
in the above minimization. 

  For the rest of this section we let $\cC$ consist of 
  non-signaling distributions for the CHSH Bell-test configuration
  with a fixed input distribution $\Prob(XY)$ where
  $\Prob(xy)>0$ for all $xy$.  As explained 
  in the previous section, $\cC$ is a convex polytope with $24$ extreme points.  
  Considering the argument below Eq.~\eqref{eq:analytic_g2}, the number of 
  terms in the sum of Eq.~\eqref{eq:analytic_g2} is at most $24$. 
  As in the previous section, we can divide the $24$ extreme points into the two
  classes consisting of the $16$ deterministic local realistic distributions 
  $\sigma_{\text{LR}_i}$, $i=1,2,...,16$, and the $8$
  variations of the PR box $\sigma_{\text{PR}_j}$, $j=1,2,...,8$. Because the 
  $\sigma_{\text{LR}_{i}}$ are deterministic conditional on the inputs, if
  $\sigma_{e}=\sigma_{\text{LR}_{i}}$ then the conditional entropy satisfies 
  $H_{\sigma_{\text{LR}_i}}(AB|XY,E=e)=0$.  For each PR box $\sigma_{\text{PR}_j}$, 
  the conditional probabilities $\sigma_{\text{PR}_j}(AB|XY)$ are either $0$ or $1/2$~\cite{barrett:2005}.  
  Thus, if $\sigma_{e}=\sigma_{\text{PR}_j}$, the conditional entropy
  satisfies $H_{\sigma_{\text{PR}_j}}(AB|XY, E=e)=1$. Hence, the minimization
  problem in Eq.~\eqref{eq:analytic_g2} becomes
\begin{equation}
  \begin{array}[b]{lll}
    g_0=&\textrm{Min:} & \sum_{j}\omega_{\text{PR}_j} \\ 
    &\textrm{With:}& \omega_{\text{LR}_i}, \omega_{\text{PR}_j}\geq 0, \forall i, j, \\
    & & \sum_{i}\omega_{\text{LR}_i}+\sum_{j}\omega_{\text{PR}_j}=1,\\
    & & \sum_{i}\omega_{\text{LR}_i}\sigma_{\text{LR}_i}+\sum_{j}\omega_{\text{PR}_j}\sigma_{\text{PR}_j}=\nu.
  \end{array}\label{eq:analytic_g3}
\end{equation}
We need to find the minimum total probability of PR boxes in a 
representation of the distribution $\nu$ as a convex combination of 
the $16$ local realistic distributions and the $8$ PR boxes. To help 
solve this problem, we consider the violation of the CHSH Bell 
inequality~\cite{clauser:qc1969a}. Recall that there is only one PR 
box that can violate a particular CHSH Bell inequality 
$\Exp(I_{\text{CHSH}})\leq 2$~\cite{barrett:2005}, where $I_{\text{CHSH}}$
is the CHSH Bell function
\begin{equation} \label{eq:CHSH_function}
I_{\text{CHSH}}(ABXY)=(1-2XY)(-1)^{A+B}/\Prob(XY),
\end{equation}
and $A, B, X, Y\in\{0,1\}$.  Let $\sigma_{\text{PR}_{1}}$ be the
  violating PR box.  The expectation of $I_{\text{CHSH}}$ 
according to $\sigma_{\text{PR}_{1}}$ is maximal, that is,
$\Exp_{\sigma_{\text{PR}_{1}}}(I_{\text{CHSH}})=4$.  
Without loss of generality, $\hat I = \Exp_{\nu}(I_{\text{CHSH}})>2$. 
The probability $\omega_{\text{PR}_1}$ in the convex decomposition of $\nu$
satisfies the inequality
$4\omega_{\text{PR}_1}+(1-\omega_{\text{PR}_1})2\geq \hat I$, or equivalently,
$\omega_{\text{PR}_1}\geq (\hat I-2)/2$.  Hence, according to
Eq.~\eqref{eq:analytic_g3}, we have $g_0\geq (\hat I-2)/2$.

We next show that $g_0\leq (\hat I-2)/2$. For this, we directly use the result of Ref.~\cite{bierhorst:2016}. 
According to Ref.~\cite{bierhorst:2016}, for any non-signaling distribution $\sigma(ABXY)$, if 
$\Exp_{\sigma}( I_{\text{CHSH}})>2$, then the distribution $\sigma(ABXY)$ can be decomposed as 
$\sigma(ABXY)= \omega_{\text{PR}_1} \sigma_{\text{PR}_1}+\sum_{i}\omega_{\text{LR}_i}\sigma_{\text{LR}_i}$, 
where $\omega_{\text{PR}_1}=(\Exp_{\sigma}( I_{\text{CHSH}})-2)/2$, $\omega_{\text{LR}_i}\geq 0$,  
and $\sum_{i}\omega_{\text{LR}_i}=1-\omega_{\text{PR}_1}$.
Specializing to the distribution $\nu(ABXY)$, we get that $g_0\leq (\hat I-2)/2$ for $\hat{I}>2$. 

The arguments above show that given $\hat I>2$, 
the maximum asymptotic randomness rate witnessed by PEFs is 
\begin{equation}
g_0=(\hat I-2)/2,
\label{eq:gain_pef}
\end{equation}
independent of the particular distribution $\nu$ realizing $\hat I$.  

We also numerically evaluated the maximum asymptotic rate
according to $g_0=\sup_{\beta>0} g(\beta)$ with
$g(\beta)$ given by Eq.~\eqref{eq:gain_rate} of the main text. 
The numerical results are presented in
Fig.~\ref{fig:optimality}, which are consistent with the analytic
expression in Eq.~\eqref{eq:gain_pef}.

\begin{figure}
  \begin{center}
    \includegraphics[scale=0.54,viewport=4.5cm 9cm 16.5cm 20cm]{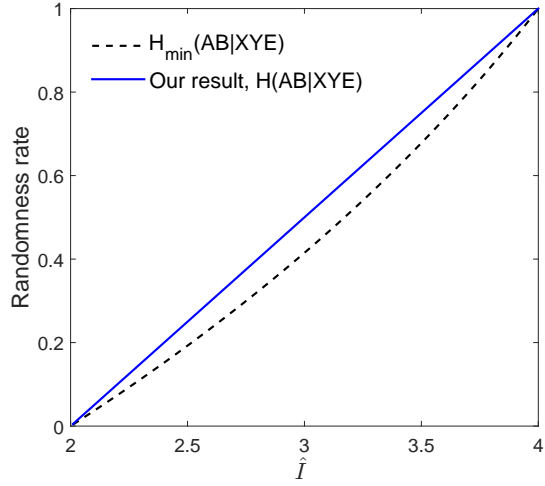}    
  \end{center}
  \caption{Asymptotic randomness rates as a function of $\hat{I}$. 
    Results according to both
    our method (the solid curve) and Refs.~\cite{pironio:2013,
      pironio:2010, fehr:2013, Acin:2012, nieto:2014, bancal:2014,
      NietoSilleras:qc2016} (the dashed curve) are shown.  Our method
    witnesses the maximum asymptotic rate $H(AB|XYE)$, which
    is the worst-case conditional entropy.}
  \label{fig:optimality}
\end{figure}

Next, we consider the quantification of the
asymptotic randomness rate by the single-trial conditional min-entropy
$H_{\min}(AB|XYE)$, which is a lower bound and is 
studied in Refs.~\cite{pironio:2013, pironio:2010, fehr:2013,
  Acin:2012, nieto:2014, bancal:2014, NietoSilleras:qc2016}.  The
single-trial conditional min-entropy is defined by
\begin{equation}
H_{\min}(AB|XYE)=-\log_2(P_{\mathrm{guess}}(AB|XYE)),
\label{eq:min-entropy}
\end{equation}
where $P_{\mathrm{guess}}(AB|XYE)$ is the average guessing probability of the output $AB$ 
given the input $XY$ and the side information $E$, as defined in Ref.~\cite{bancal:2014}. 
According to Refs.~\cite{nieto:2014, bancal:2014}, the guessing probability at $xy$ is given 
by the following maximization:
\begin{equation}
P_{\mathrm{guess}}(AB|xyE)=\max_{\omega_e, \sigma_e}\left\{\sum_{e}\omega_e \max_{ab}\sigma_e(ab|xy):\forall e, \sigma_{e}\in\cC \text{ and }\omega_{e}\geq 0, \sum_{e}\omega_{e}=1,\sum_{e}\omega_{e}\sigma_{e}=\nu \right\}.
\label{eq:analytic_guessingprob1}
\end{equation}
If a $\sigma_e$ contributing to the sum in Eq.~\eqref{eq:analytic_guessingprob1} 
is not extremal in the set $\cC$, we can replace it by a convex
combination of extremal distributions to increase the value of the
sum.  Thus, we only have to consider extremal distributions $\sigma_e$ in
the above maximization.  Applying the argument that led from 
Eq.~\eqref{eq:analytic_g2} to Eq.~\eqref{eq:analytic_g3}, we obtain 
\begin{equation}
  \begin{array}[b]{lll}
    P_{\mathrm{guess}}(AB|xyE)=&\textrm{Max:} & \sum_{i}\omega_{\text{LR}_i}+\frac{1}{2}\sum_{j}\omega_{\text{PR}_j} \\ 
    & \textrm{With:}& \omega_{\text{LR}_i}, \omega_{\text{PR}_j}\geq 0, \forall i, j, \\
    & & \sum_{i}\omega_{\text{LR}_i}+\sum_{j}\omega_{\text{PR}_j}=1,\\
    & & \sum_{i}\omega_{\text{LR}_i}\sigma_{\text{LR}_i}+\sum_{j}\omega_{\text{PR}_j}\sigma_{\text{PR}_j}=\nu.
  \end{array}\label{eq:analytic_guessingprob2}
\end{equation}
Since
$\sum_{i}\omega_{\text{LR}_i}+\frac{1}{2}\sum_{j}\omega_{\text{PR}_j}=1-\frac{1}{2}\sum_{j}\omega_{\text{PR}_j}$,
we only need to minimize the total probability of PR boxes
$\sum_{j}\omega_{\text{PR}_j}$ in the convex decomposition of the
distribution $\nu$. From the derivation of $g_0$ that
  gave Eq.~\eqref{eq:gain_pef}, we conclude that
$\min \left(\sum_{j}\omega_{\text{PR}_j}\right)=(\hat
I-2)/2$ for $\hat{I}>2$. Therefore $P_{\mathrm{guess}}(AB|xyE)=(6-\hat I)/4$ regardless
of the particular input $xy$. Furthermore, the specific convex
decomposition over $E$ that achieves the maximum in 
Eq.~\eqref{eq:analytic_guessingprob2} is the same for all the possible 
inputs. Hence we also have $P_{\mathrm{guess}}(AB|XYE)=(6-\hat I)/4$
independent of the input distribution.  Therefore the
single-trial conditional min-entropy is
\begin{equation}
H_{\min}(AB|XYE)=-\log_2((6-\hat I)/4),
\end{equation}
which is plotted in Fig.~\ref{fig:optimality}. 

The results of this section are summarized in the following theorem:
\begin{theorem}\label{thm:gain_rate}
  Suppose that the trial model $\cC$ consists of non-signaling distributions
  with a fixed input distribution $\Prob(XY)$ where $\Prob(xy)>0$ for all $xy$. 
  For any $\nu\in\cC$,  both the maximum asymptotic randomness rate $g_0$ 
  witnessed by PEFs and the single-trial conditional min-entropy
  $H_{\min}(AB|XYE)$ depend only on $\hat{I}=\Exp_{\nu}(I_{\text{CHSH}})>2$ and 
  are given by $g_0=(\hat I-2)/2$ and $H_{\min}(AB|XYE)=-\log_2((6-\hat I)/4)$.
\end{theorem}

\section{Entropy Accumulation} 
\label{subsec:eat}
Consider DIRG in the CHSH Bell-test configuration. In this section, the input
distribution $\Prob(XY)$ at each trial is assumed to be uniform. Define the winning
probability at a trial by $\hat{\omega}=1/2+\hat{I}/8$ where
$\hat{I}=\Exp_{\nu}(I_{\text{CHSH}})$ with $\nu$ the 
  distribution of trial results. Entropy accumulation~\cite{arnon-friedman:2018} 
  is a framework for estimating (quantum) conditional min-entropy 
  with respect to quantum side information and can be applied 
  to the CHSH Bell-test configuration. The following theorem from
  Ref.~\cite{arnon-friedman:2018} implements the framework: 
  
\begin{theorem} \label{thm:EAT} Let $(2+\sqrt{2})/4\geq\omega_{\mathrm{exp}},
  p_t\geq 3/4$, and $1\geq \kappa, \epsilon>0$.   Suppose that
  after $n$ trials the joint quantum state of the inputs $\Sfnt{XY}$,
  the outputs $\Sfnt{AB}$ and the quantum side information $E$ is
  $\rho$.  Define $\{\phi\}$ to be the event that the experimentally 
  observed winning probability is higher than or equal to $\omega_{\mathrm{exp}}$,
  and suppose that $\kappa\leq\Prob_{\rho}(\phi)$. 
  Denote the joint quantum state conditional on
  $\{\phi\}$ by $\rho_{|\phi}$. Then the (quantum) smooth conditional
  min-entropy evaluated at $\rho_{|\phi}$ satisfies
\begin{equation*}\label{eq:eat}
H_{\min}^{\epsilon}(\Sfnt{AB}|\Sfnt{XY}E)_{\rho_{|\phi}} > n \eta(p_t, \omega_{\mathrm{exp}}, n, \epsilon, \kappa), 
\end{equation*}
where $\eta$ is defined by 
\begin{align} \label{eq:eat_rate}
&g(p) =  \begin{cases}
			 1 - h\left( \frac{1}{2} + \frac{1}{2}\sqrt{16p(p-1)+3}  \right)&  p\in\left[3/4,(2+\sqrt{2})/4\right] \\
			1 & p \in\left[(2+\sqrt{2})/4,1\right]\;,
			\end{cases}\notag\\
	&f_{\min}\left(p_t,p\right) = \begin{cases}
	g\left(p\right)&  p \leq p_t \;  \\
	\frac{\mathrm{d}}{\mathrm{d}p} g(p)\big|_{p_t} p+ \Big( g(p_t) -\frac{\mathrm{d}}{\mathrm{d}p} g(p)\big|_{p_t} p_t\Big)& p> p_t\;,
	\end{cases} \nonumber \\
	& v(p_t,\epsilon, \kappa)=2\left( \log_2 13 + \frac{\mathrm{d}}{\mathrm{d}p} g(p)\big|_{p_t}\right)\sqrt{1-2 \log_2 (\epsilon\kappa)}\;, \notag \\ 
	&\eta(p_t,p,n,\epsilon, \kappa) =  f_{\min}\left(p_t, p\right) - \frac{1}{\sqrt{n}}v(p_t,\epsilon, \kappa)\;,\notag
\end{align}
with $h(x)=-x\log_{2}(x)-(1-x)\log_{2}(1-x)$ be the binary entropy function.
\end{theorem}

The function $f_{\min}$ in the theorem is referred to as a min-tradeoff function.
The parameter $p_t$ in the theorem is free, and can be optimized
over its range before running the protocol based on the chosen parameters $n$,
$\omega_{\textrm{exp}}$, $\epsilon$ and $\kappa$. So the optimal
entropy rate is $\eta_{\textrm{opt}}(\omega_{\textrm{exp}},
n,\epsilon, \kappa)=\max_{p_t}\eta(p_t,\omega_{\textrm{exp}},
n,\epsilon, \kappa)$.

According to Thm.~\ref{thm:EAT}, in order to certify $b$ bits of entropy given $\omega_{\mathrm{exp}}$, $\epsilon$ 
and $\kappa$, we need that $n \eta(p_t, \omega_{\mathrm{exp}}, n, \epsilon, \kappa)\geq b$. Equivalently, $n\geq n_{\mathrm{EAT},b}(p_t)$ where 
\begin{equation} \label{eq:EAT_trial_bound}
n_{\mathrm{EAT},b}(p_t)=\Bigg(\frac{ v(p_t,\epsilon, \kappa)+\sqrt{ v(p_t,\epsilon, \kappa)^2+4bf_{\min}\left(p_t,\omega_{\mathrm{exp}}\right)}}{2f_{\min}\left(p_t,\omega_{\mathrm{exp}}\right)}\Bigg)^2.
\end{equation}
Including the optimization over $p_{t}$ gives the minimum number of identical trials required:
\begin{equation} \label{eq:opt_EAT_trial_bound}
n_{\mathrm{EAT},b}=\min_{3/4\leq p_t\leq (2+\sqrt{2})/4}n_{\mathrm{EAT},b}(p_t). 
\end{equation}
To compute $n_{\textrm{EAT},b}$, we set the parameter $\omega_{\mathrm{exp}}$ to the 
winning probability $\hat{\omega}$ according to the distribution $\nu$
 of trial results in a stable experiment. 

We finish with several remarks on the comparison between entropy accumulation and probability estimation.  
First, Thm.~\ref{thm:EAT} based on entropy accumulation holds with respect to quantum 
side information, while Cor.~\ref{cor:smooth_min_entropy_bound} 
(Thm.~\ref{thm:smooth_min_entropy_bound} in the main text) %
based on probability estimation holds with respect to classical side information.
Second, in principle both entropy accumulation and
  probability estimation can witness asymptotically tight bounds on the smooth
conditional min-entropies with respect to the assumed side information.  
Entropy accumulation can witness the maximum asymptotic entropy rate with
respect to quantum side information, \emph{if} an optimal min-tradeoff
function is available.  However, it is unknown how to obtain such
 min-tradeoff functions. In particular, the min-tradeoff
function $f_{\min}\left(p,p_t\right)$ is not optimal for the CHSH Bell-test 
configuration considered here. A min-tradeoff function is required to be a
  bound on the single-trial conditional von Neumann entropy $H(AB|XYE)$.
  That $f_{\min}\left(p,p_{t}\right)$ is not optimal is due to the
  following: 1) $f_{\min}\left(p,p_t\right)$ is designed according
  to a bound on the single-trial conditional von Neumann entropy
  $H(A|XYE)$ derived in Refs.~\cite{acin:2007, pironio:2009}.  A tight
  bound on $H(A|XYE)$ is generally not a tight bound on $H(AB|XYE)$. 
  2) The bound on $H(A|XYE)$ derived in Refs.~\cite{acin:2007,
  pironio:2009} is tight \emph{if} the only information available is
the winning probability. However, in practice one can access the full
measurement statistics rather than just the winning
probability. In contrast to entropy accumulation,  probability estimation is an effective 
method for approaching the maximum asymptotic entropy rate (with respect to
classical side information) considering the full measurement
statistics and the model constraints. In general, the maximum rate 
with respect to quantum side information is lower than that with respect to
classical side information, as accessing quantum side information
corresponds to a more powerful attack.  Third and as demonstrated
  in the main text, probability estimation performs significantly better with finite data.

\end{document}

%% file: xbytm_defs.tex
\providecommand{\ignore}[1]{}

\newif\ifcmnt
\cmnttrue
\ifdefined\cmntsoff\cmntfalse\fi

\ifcmnt
    \providecommand{\aucmnt}[1]{#1}

\else
    \providecommand{\aucmnt}[1]{}

\fi

\newcommand{\Prob}{\mathbb{P}}
\newcommand{\Exp}{\mathbb{E}}

\newcommand{\Probv}{\mathbb{P}}

\newcommand{\TV}{\mathrm{TV}}

\newcommand{\UPE}{\mathrm{UPE}}

\newcommand{\knuth}[1]{\left\llbracket #1 \right\rrbracket}

\newcommand{\Rng}{\mathrm{Rng}}

\newcommand{\Sfnt}[1]{\mathbf{#1}} 
\newcommand{\Pfnt}[1]{\textsf{#1}} 

\newcommand{\cC}{\mathcal{C}}
\newcommand{\cD}{\mathcal{D}}
\newcommand{\cE}{\mathcal{E}}

\newcommand{\cH}{\mathcal{H}}

\newcommand{\cN}{\mathcal{N}}

\newcommand{\cP}{\mathcal{P}}

\newcommand{\cT}{\mathcal{T}}

\def\one{{\mathchoice {\rm 1\mskip-4mu l} {\rm 1\mskip-4mu l} {\rm
1\mskip-4.5mu l} {\rm 1\mskip-5mu l}}}

\newtheorem{theorem}{Theorem}
\newtheorem*{theorem*}{Theorem}
\newtheorem{lemma}[theorem]{Lemma}
\newtheorem*{lemma*}{Lemma}
\newtheorem{corollary}[theorem]{Corollary}
\newtheorem{definition}[theorem]{Definition}
\newtheorem*{definition*}{Definition}


%% file: probability_estimation.bbl
%